\documentclass[a4paper,reqno]{amsart}
\usepackage{a4wide,color}
\usepackage{amsmath,amssymb,amsthm,mathrsfs}
\usepackage[american]{babel}
\usepackage[pdftex]{graphicx}

\newenvironment{myenumerate}{
	
	\begin{enumerate}
		\setlength{\itemsep}{0em}\setlength{\parsep}{0em}%
		\setlength{\topsep}{0em}\setlength{\parskip}{0em}%
	}
{ 	\end{enumerate} }

\numberwithin{equation}{section}

\def\R{\mathbb{R}}
\def\N{\mathbb{N}}
\def\B{\mathcal{B}}
\def\Lebesgue{\mathcal{L}}

\def\I{\mathcal{I}}
\def\apppf{\sharp}
\def\lint{\displaystyle\int\limits}

\DeclareMathOperator{\supp}{supp}

\DeclareMathOperator{\Int}{Int}
\DeclareMathOperator{\tr}{tr}

\newtheorem{assumption}{Assumption}
\newtheorem{theorem}{Theorem}
\newtheorem{proposition}[theorem]{Proposition}
\newtheorem{lemma}[theorem]{Lemma}
\theoremstyle{remark}\newtheorem*{remark}{Remark}

\title{Time-evolving measures and macroscopic modeling of pedestrian flow}

\author{Benedetto Piccoli}
\address{Istituto per le Applicazioni del Calcolo ``Mauro Picone'' \\
		Consiglio Nazionale delle Ricerche \\
		Viale del Policlinico 137, 00161 Roma, Italy}
\email{b.piccoli@iac.cnr.it}

\author{Andrea Tosin}
\address{Istituto per le Applicazioni del Calcolo ``Mauro Picone'' \\
		Consiglio Nazionale delle Ricerche \\
		Viale del Policlinico 137, 00161 Roma, Italy}
\email{a.tosin@iac.cnr.it}


\begin{document}

\begin{abstract}
This paper deals with the early results of a new model of pedestrian flow, conceived within a measure-theoretical framework. The modeling approach consists in a discrete-time Eulerian macroscopic representation of the system via a family of measures which, pushed forward by some motion mappings, provide an estimate of the space occupancy by pedestrians at successive time steps. From the modeling point of view, this setting is particularly suitable to treat nonlocal interactions among pedestrians, obstacles, and wall boundary conditions. In addition, analysis and numerical approximation of the resulting mathematical structures, which is the main target of this work, follow more easily and straightforwardly than in case of standard hyperbolic conservation laws, also used in the specialized literature by some Authors to address analogous problems.
\end{abstract}

\maketitle

\section{Introduction}
In the last decades, the modeling of crowd motion and pedestrian flow has drawn the attention of applied mathematicians, because of an increasing number of applications, in engineering and social sciences, dealing with this or similar complex systems. First studies on pedestrian behavior, dating back to the sixties, were aimed at providing guidelines for the design of walkway infrastructures in urban areas, and yet nowadays crowd-structure interaction is the object of scientific investigations (see e.g., Venuti \textit{et al.} \cite{MR2284944} and the main references therein). More recently, management and optimization of pedestrian fluxes in large and crowded environments, like airports, stations, shopping centers, stadiums, have motivated new efforts toward more accurate experimental investigations (see e.g., Helbing \emph{et al.} \cite{HeJoAA,HeMoFaBo}) and a more targeted mathematical modeling of human walking attitudes in bounded areas, with special emphasis on the concepts of ``walking program'' of pedestrians and of pedestrian-pedestrian, pedestrian-obstacle interactions.

The existing literature on mathematical modeling of human crowds is especially concerned with models at the microscopic scale, in which pedestrians are described individually in their motion by a set of ordinary differential equations. Problems are usually set in two-dimensional domains delimiting the walking area under consideration, and the presence of obstacles within the domain, as well as of targets that pedestrians aim at, is possibly taken into account. The basic modeling framework relies on classical Newtonian laws of point mechanics: Pedestrians are assimilated to rigid disks of fixed, possibly different, radii, with a velocity, or alternatively an acceleration, that takes into account both the desired direction of motion and a superimposed perturbation due to the presence of other pedestrians in the surroundings
(Helbing and Johansson \cite{HeJo}, Maury and Venel \cite{MaVe}). In practice, pedestrians are assumed to move mainly according to a \emph{desired velocity}, i.e., the velocity they would have in the absence of other people, so as to get
a desired target and avoid at the same time intermediate obstacles. The desired velocity is however modified by the necessity to fulfill some local geometrical constraints of maximal congestion. More complicated strategies to determine the direction of motion may involve the minimization of some `walking cost', that each pedestrian evaluates on the basis of the current configuration of the system in her/his neighborhood and of her/his prediction about the motion of other pedestrians (Hoogendoorn and Bovy \cite{HoBo2}).

A different approach to the problem, however not yet as much developed in the specialized literature, consists in using partial differential equations at the macroscopic scale, that is in describing the evolution in time and space of pedestrian density rather than following each subject individually. The starting point is the analogy with classical fluid dynamics: An Eulerian point of view is adopted, appealing to the conservation of mass of pedestrians supplemented by either suitable closure relations linking the velocity of the latter to their density (see e.g., Hughes \cite{Hu}) or an analogous balance law for the momentum (see e.g., Bellomo and Dogb\'e \cite{MR2438218}). Again, typical guidelines in
devising this kind of models are the concepts of preferred direction of motion and discomfort at high densities. In the framework of scalar conservation laws, a macroscopic one-dimensional model has been proposed by Colombo and Rosini
\cite{MR2158218}, resorting to some common ideas to vehicular traffic modeling (e.g., fundamental diagrams), with the specific aim of describing the transition from normal to panic conditions.

Microscopic models allow an accurate description of the behavior of each single agent of the system, but they usually require a large number of ordinary differential equations to be dealt with simultaneously, as many as the number of pedestrians considered in the model. This often makes it difficult, for practical purposes, to recover macroscopic information on the system from the knowledge of the microscopic state, hence to investigate the relevant global features possibly also in connection with control and optimization problems. On the other hand, macroscopic models existing in the literature heavily rely on the fluid dynamics analogy, hence they use mathematical tools proper of hyperbolic conservation laws, which however might not be the most appropriate ones to address the behavior of pedestrians. For instance, hyperbolic equations require the imposition of boundary conditions in a technically tricky way, which may not correspond to the actual modeling needs. Moreover, crowd motion is definitely not one-dimensional: A two-dimensional description is inevitable in order to catch the fundamental aspects of the system, but it is known that multidimensional hyperbolic equations generate additional analytical and numerical difficulties, leading to an increase of technical complexity in handling the final product.

Bearing all above in mind, we propose to adopt a different macroscopic point of view, based on a measure-theoretical framework which has recently been introduced by Canuto \emph{et al.} \cite{CaFaTi} for coordination problems
(\emph{rendez-vous}) of multiagent systems.

Given a two-dimensional spatial domain $\Omega\subset\R^2$, possibly containing obstacles understood as internal boundaries, the basic idea is to describe the space occupancy by pedestrians at time $n$ via a measure $\mu_n$, that is, roughly speaking, a mapping that to each subset $E$ of $\Omega$ associates a real nonnegative number $\mu_n(E)$ representing an estimate, in macroscopic averaged terms, of the amount of people contained in $E$. The whole mass of
pedestrians is then obtained as $\mu_n(\Omega)$, and its conservation in time is possibly achieved by requiring $\mu_n(\Omega)=\mu_0(\Omega)$ for all $n>0$. If $v_n=v_n(x)$ is the velocity field at time $n$ of pedestrians located at a point $x\in\Omega$, the dynamics of the system is described by a family of mappings $\{\gamma_n\}_{n\geq 0}$ such that $\gamma_n(x)-x=v_n(x)\Delta{t}$, which entails the following evolution for the measure $\mu_n$ (\emph{push forward}):
$$ \mu_{n+1}(E)=\mu_n(\gamma_n^{-1}(E)), \qquad \forall\,E\subseteq\Omega. $$
This is nothing but a formal mathematical statement of the simple idea that the amount of people contained in
a spatial region $E$ at time $n>0$ is related to the analogous amount at the initial time $n=0$ along the trajectories of the motion of pedestrians themselves. In this context, we recall that a measure theoretical approach for mass transportation problems has been proposed also by Buttazzo \emph{et al.} \cite{BuJiOu}, focusing in particular on concentration and congestion effects. Specifically, the Authors claim that the latter may be applied to the modeling of crowd dynamics in panic situations.

The construction of the velocity $v_n$ is the main modeling task. The framework we use easily allows to duly incorporate the concept of desired velocity, as well as to model interactions among pedestrians, possibly accounting for averaged non-local effects due to the surrounding crowding. By naturally tracing, for fixed $x\in\Omega$, the displacements of pedestrians, the mapping $\gamma_n(x)$ describes the system under an essentially Lagrangian point of view, which is, in a sense, a more natural way to look at the motion of the agents. On the other hand, the push forward of the measure $\mu_n$ finally refers to an Eulerian handling of the system, and takes thus advantage of a fixed geometry without the need for resorting to the concept of evolving reference configuration.

The paper is organized into three more sections that follow this Introduction. In Sect. \ref{sect:theory} the modeling framework by time-evolving measures is introduced for generic $d$-dimensional continuous systems. After a first parallelism with the classical mass conservation equation of continuum mechanics, the theory focuses on discrete time push forward, which is shown to represent an explicit time discretization of a conservation law for a family of continuous-time-evolving measures. Well posedness and spatial approximation of the discrete time model are addressed. In addition, a computational scheme for its numerical treatment is derived, with a related convergence and error analysis. In Sect. \ref{sect:modeling} the above theoretical framework is specifically applied to pedestrian flow. From the modeling point of view, this setting proves to be useful in extending the idea of macroscopic mass to cases, like in fact human crowds, in which mass and mass density may not be conceptually defined in the straightforward sense of continuum mechanics. It may be questioned that this problem is actually common to many other classical and nonclassical systems, ranging from gas dynamics to vehicular traffic, for which, despite the essential granular nature of the matter under consideration, the continuum approach has been successfully used. Nevertheless, in the case of pedestrian flow the measure theoretical approach allows to deal with nonlocal interactions among pedestrians, obstacles and more generally wall boundary conditions, and numerical approximations more effectively than standard hyperbolic conservation laws. Finally, Sect. \ref{sect:conclusions} draws some conclusions and briefly sketches research perspectives.

\section{Time-evolving measures theoretical framework}	\label{sect:theory}
Many real world systems can be mathematically described at the macroscopic scale, usually invoking conservation or balance laws of some gross quantities. As a matter of fact, in modern applications dealing with life and social sciences, e.g., biological systems, vehicular traffic, pedestrian flow, often the only conservation principle which is reasonable to resort to is the conservation of mass, because the evolution of the system involves complex behavioral aspects which need not preserve other classical mechanical quantities, like linear momentum or energy.

In the physical space $\R^d$ ($d=1,\,2,\,3$ for applications), the \emph{mass} can be viewed as a \emph{Radon positive measure} $\mu:\B(\R^d)\to\R_+$, i.e., a mapping from the Borel $\sigma$-algebra $\B(\R^d)$ to the real nonnegative line, whose countable additivity furnishes the mathematical counterpart of the principle of additivity of the mass. The Lebesgue measure $\Lebesgue^d$ on $\R^d$ formalizes instead the concept of \emph{volume}. If one assumes the \emph{continuum hypothesis}, stating that the mass $\mu$ is absolutely continuous with respect to $\Lebesgue^d$ (written $\mu\ll\Lebesgue^d$), namely that every body with zero volume has also zero mass, then Radon-Nikodym theorem implies the existence of a function $\rho\in L^1_{loc}(\R^d)$, $\rho\geq 0$ a.e. in $\R^d$, called the \emph{mass density}, such that $d\mu=\rho\,dx$.

In continuum mechanics, for a mass density $\rho_t$ evolving in time\footnote{Throughout this section, the subscript $t$ does not denote the partial derivative with respect to time. We use the notation $\rho_t(x)$, instead of the more classical $\rho(t,\,x)$ (and similarly for other functions), to emphasize the idea that the density should be regarded as a function of $x\in\R^d$ parameterized by $t>0$.} the principle of conservation of the mass is stated in Lagrangian form as
\begin{equation}
	\frac{d}{dt}\lint_{\Gamma_t(E)}\rho_t(x)\,dx=0,\qquad
		\forall\,E\in\B(\R^d),
	\label{eq:mass-Lag}
\end{equation}
where $\Gamma_t$ is the \emph{motion mapping}, which describes the motion of the points of $\R^d$ under the action of a certain velocity field $V_t$ as
$$ \dot{\Gamma}_t(x)=V_t(\Gamma_t(x)), $$
supplemented by the further condition $\Gamma_0(x)=x$. In practice, $\Gamma_t(x)$ is the position occupied at time $t$ by the point initially located at $x$. Using Reynolds theorem, Eq. \eqref{eq:mass-Lag} can be formally converted in Eulerian form and rewritten as the well-known conservation law for the density $\rho_t$:
\begin{equation}
	\frac{\partial\rho_t}{\partial t}+\nabla\cdot(\rho_tV_t)=0.
	\label{eq:mass-Eul}
\end{equation}
Equivalence between Eqs. \eqref{eq:mass-Lag} and \eqref{eq:mass-Eul}, however, holds for smooth $\rho_t$ and $V_t$ only, as passing from Eq. \eqref{eq:mass-Lag} to Eq. \eqref{eq:mass-Eul} requires some manipulations which are not valid if
the involved functions are not smooth.

As a natural generalization of Eq. \eqref{eq:mass-Eul} to the case in which a mass density need not be present, one can assume that a family of time-evolving measures $\mu_t:\B(\R^d)\to\R_+$, $t>0$, is given, satisfying
\begin{equation}
	\frac{\partial\mu_t}{\partial t}+\nabla\cdot(\mu_t V_t)=0.
	\label{eq:mass-meas}
\end{equation}
This partial differential equation has to be meant in the sense of measures as follows: For every infinitely differentiable test function $\eta$ with compact support in $\R^d$, i.e., $\eta\in C^\infty_0(\R^d)$,
\begin{equation}
	\frac{d}{dt}\lint_{\R^d}\eta(x)\,d\mu_t(x)=\lint_{\R^d}\nabla{\eta}(x)\cdot V_t(x)\,d\mu_t(x).
	\label{eq:mass-meas-weak}
\end{equation}
In particular, a family of measures $\{\mu_t\}_{t>0}$ is said to be a solution to Eq. \eqref{eq:mass-meas} if for all $\eta\in C^\infty_0(\R^d)$ the mapping
$$ t\mapsto\lint_{\R^d}\eta(x)\,d\mu_t(x) $$
is absolutely continuous and satisfies Eq. \eqref{eq:mass-meas-weak}. Notice that a basic requirement for the right-hand side of Eq. \eqref{eq:mass-meas-weak} to be well defined is $V_t\in(L^1(\R^d,\,\mu_t))^d$ for all $t>0$, i.e.,
$$ \lint_{\R^d}\vert V_t(x)\vert\,d\mu_t(x)<+\infty, \qquad \forall\,t>0. $$

Let us denote by $\supp{\mu_t}$ the \emph{support} of the measure $\mu_t$, and let us assume that an open set $\Omega\subset\R^d$ exists such that $\supp{\mu_t}\subset\subset\Omega$ for all $t>0$. We can then rewrite Eq. \eqref{eq:mass-meas-weak} by integrating on $\Omega$ only:
\begin{equation}
	\frac{d}{dt}\lint_{\Omega}\eta(x)\,d\mu_t(x)=\lint_{\Omega}\nabla{\eta}(x)\cdot V_t(x)\,d\mu_t(x),
		\qquad \forall\,\eta\in C^\infty_0(\Omega).
	\label{eq:mass-meas-weak-2}
\end{equation}
In addition, there exist an open set $U\subset\Omega$ and a test function $\eta_U\in C^\infty_0(\Omega)$ such that $\supp{\mu_t}\subset U$ for each $t>0$, $\eta_U\equiv 1$ on $\supp{\mu_t}$, and $\eta_U\equiv 0$ on $\Omega\setminus U$, whence Eq. \eqref{eq:mass-meas-weak-2} with the choice $\eta=\eta_U$ gives
$\frac{d}{dt}\mu_t(\supp{\mu_t})=0$. However, $\mu_t(\supp{\mu_t})=\mu_t(\Omega)$ for all $t>0$, thus finally
\begin{equation}
	\frac{d}{dt}\mu_t(\Omega)=0,
	\label{eq:consmass_cont}
\end{equation}
which states that no matter is flowing through the boundary $\partial\Omega$ at any time: The measure of $\Omega$ is not varying, hence the whole mass remains concentrated within it. Condition $\supp{\mu_t}\subset\subset\Omega$ for all $t>0$
compares therefore with a classical homogeneous Neumann boundary condition for the PDE \eqref{eq:mass-Eul}.

\subsection{Push forward and time discretization}	\label{subsect:pushfwd_timedisc}
When dealing with a discrete time evolution of the system, one can fix a time step $\Delta{t}>0$ and trace the mass by a sequence of positive measures $\{\mu_n\}_{n>0}$, $\mu_n:\B(\R^d)\to\R_+$ each $n$, via the following recurrence relation (\emph{push forward}):
\begin{equation}
	\mu_{n+1}=\gamma_n\#\mu_n,
	\label{eq:pushfwd}
\end{equation}
where $\gamma_n:\R^d\to\R^d$ is the \emph{one-step motion mapping} (briefly termed \emph{motion mapping} in the sequel for simplicity). More specifically,
\begin{equation}
	\gamma_n(x)=x+v_n(x)\Delta{t},
	\label{eq:gamman}
\end{equation}
with $v_n:\R^d\to\R^d$ a velocity field, so that $\gamma_n(x)$ is the position at the $(n+1)$-th time step of the point which at the $n$-th time step is located in $x$.

The push forward \eqref{eq:pushfwd} has to be understood formally as
\begin{equation}
	\mu_{n+1}(E)=\mu_n(\gamma_n^{-1}(E)), \qquad \forall\,E\in\B(\R^d),
	\label{eq:pushfwd_2}
\end{equation}
which shows that $\mu_{n+1}$ is unaffected by the values that $\gamma_n$ possibly takes outside $\supp{\mu_n}$. On the other hand, $\supp{\mu_{n+1}}=\gamma_n(\supp{\mu_n})$, hence if there exists a set $\Omega\subseteq\R^d$ such that $\gamma_n(\Omega)\subseteq\Omega$ for all $n$, and furthermore an initial measure $\mu_0:\B(\R^d)\to\R_+$ is prescribed with $\supp{\mu_0}\subseteq\Omega$, then $\supp{\mu_n}\subseteq\Omega$ for all $n>0$, whence we get, analogously to Eq. \eqref{eq:consmass_cont}, the conservation of the mass of $\Omega$:
$$ \mu_{n}(\Omega)=\mu_0(\Omega), \qquad \forall\,n>0. $$

The definition of push forward given by Eq. \eqref{eq:pushfwd_2} is equivalent to
\begin{equation}
	\lint_{\R^d}\eta(x)\,d\mu_{n+1}(x)=\lint_{\R^d}\eta(\gamma_n(x))\,d\mu_n(x)
	\label{eq:pushfwd_3}
\end{equation}
for every bounded and Borel function $\eta:\R^d\to\R$. In particular, if $\eta\in C^\infty_0(\R^d)$ we can expand
$$ \eta(\gamma_n(x))=\eta(x+v_n(x)\Delta{t})=\eta(x)+\Delta{t}\nabla{\eta}(x)\cdot v_n(x)+O(\Delta{t}^2) $$
and then plug this expression into Eq. \eqref{eq:pushfwd_3} to get
$$ \frac{1}{\Delta{t}}\left[\lint_{\R^d}\eta(x)\,d\mu_{n+1}-\lint_{\R^d}\eta(x)\,d\mu_n\right]=
	\lint_{\R^d}\nabla{\eta}(x)\cdot v_n(x)\,d\mu_n+O(\Delta{t}), $$
which, at least for $\mu_n$-uniformly bounded $v_n$, can be viewed as an explicit time discretization of Eq. \eqref{eq:mass-meas-weak} up to identifying $v_n(x)=V_{n\Delta{t}}(x)$. Again, if $\supp{\mu_n}\subseteq\Omega$ for all
$n>0$ then Eq. \eqref{eq:pushfwd_3} can be rewritten by integrating on $\Omega$ only:
\begin{equation}
	\lint_\Omega\eta(x)\,d\mu_{n+1}(x)=\lint_\Omega\eta(\gamma_n(x))\,d\mu_n(x),
	\label{eq:pushfwd_4}
\end{equation}
and the same Taylor expansion performed above provides now an explicit time discretization of Eq. \eqref{eq:mass-meas-weak-2} for test functions $\eta\in C^\infty_0(\Omega)$.

\begin{proposition}	\label{prop:mass-Lag-mu}
The push forward \eqref{eq:pushfwd} of the measure $\mu_n$ is the direct time discretization of the conservation law
\begin{equation}
	\frac{d}{dt}\mu_t(\Gamma_t(E))=0, \qquad \forall\,E\in\B(\R^d).
	\label{eq:mass-Lag-mu}
\end{equation}
\end{proposition}
\begin{proof}
To see this, discretize the motion mapping $\Gamma_t$ by the family of one-step motion mappings $\{\gamma_n\}_{n>0}$ as $\Gamma_{n+1}=\gamma_n\circ\dots\circ\gamma_0$, then approximate the time derivative in Eq. \eqref{eq:mass-Lag-mu} above as
$$ \frac{\mu_{n+1}(\Gamma_{n+1}(E))-\mu_n(\Gamma_n(E))}{\Delta{t}}=0. $$
Setting $\tilde{E}:=\Gamma_{n+1}(E)$, and consequently $\Gamma_n(E)=\gamma_n^{-1}(\tilde{E})$, yields
$$ \mu_{n+1}(\tilde{E})=\mu_n(\gamma_n^{-1}(\tilde{E})), $$
whence we recover formally the relation $\mu_{n+1}=\gamma_n\#\mu_n$ (cf. also Eq. \eqref{eq:pushfwd_2}) or, equivalently, Eqs. \eqref{eq:pushfwd_3}, \eqref{eq:pushfwd_4}.
\end{proof}

\begin{remark}
The Lagrangian mass conservation law \eqref{eq:mass-Lag} is a particular case of Eq. \eqref{eq:mass-Lag-mu} for the measure $d\mu_t=\rho_t\,dx$. It is worth noticing that, unlike classical procedures to derive the pointwise Eulerian mass conservation equation \eqref{eq:mass-Eul}, or analogously one of its weak forms \eqref{eq:mass-meas-weak}, \eqref{eq:mass-meas-weak-2}, from Eq. \eqref{eq:mass-Lag}, the time discretization allows very few regularity of the involved fields in order to attain to Eqs. \eqref{eq:pushfwd_3}, \eqref{eq:pushfwd_4} from Eq. \eqref{eq:mass-Lag-mu}.
\end{remark}

In order to deal with problems posed in a fixed bounded domain $\Omega\subset\R^d$, from now on we concentrate on motion mappings $\gamma_n:\Omega\to\Omega$ such that $\gamma_n(\Omega)\subseteq\Omega$. This way, given an initial measure $\mu_0$ supported in $\Omega$, all measures $\mu_n$ deduced from the recurrence relation \eqref{eq:pushfwd} are supported in $\Omega$ as well, hence we can simply think of them as defined on the Borel $\sigma$-algebra $\B(\Omega)$ and refer directly to Eq. \eqref{eq:pushfwd_4} whenever necessary. Notice that the measure of $\Omega$ is conserved in this case, indeed $\mu_{n+1}(\Omega)=\mu_n(\gamma_n^{-1}(\Omega))=\mu_n(\Omega)$ and the claim easily follows by induction. In particular, if $\mu_0(\Omega)<+\infty$ then $\Omega$ will have a finite measure for all successive times $n>0$.

We begin by stating some basic properties of the motion mappings.
\begin{assumption}	\label{ass:gamma}
For each $n\in\N$, we assume that the motion mapping $\gamma_n$
\begin{enumerate}
\item is Borel, i.e., $\gamma_n^{-1}(E)\in\B(\Omega)$ for all $E\in\B(\Omega)$;
\item \label{ass:gamma-Cn}
satisfies
\begin{equation}
	\Lebesgue^d(\gamma_n^{-1}(E))\leq C\Lebesgue^d(E),\qquad\forall\,E\in\B(\Omega),
	\label{eq:ass-gamma-Cn}
\end{equation}
for a certain constant $C>0$.
\end{enumerate}
\end{assumption}
Let us briefly comment on property \eqref{ass:gamma-Cn} of Assumption \ref{ass:gamma}. Given a Borel set $E\subseteq\Omega$, Eq. \eqref{eq:ass-gamma-Cn} allows to control the Lebesgue measure of the inverse image of $E$ by the Lebesgue measure of $E$ itself. From a different point of view, it requires that $\gamma_n$ does not map Lebesgue non-negligible subsets of $\Omega$ into Lebesgue negligible sets. Indeed, if $A\subseteq\Omega$ is such that $\Lebesgue^d(A)>0$, it is impossible to have $\Lebesgue^d(\gamma_n(A))=0$, for Eq. \eqref{eq:ass-gamma-Cn} implies $\Lebesgue^d(\gamma_n(A))\geq C^{-1}\Lebesgue^d(A)>0$ (take formally $E=\gamma_n(A)$).

Property \eqref{ass:gamma-Cn} is for instance satisfied if the following holds
true:
\begin{equation}
	\vert x-y\vert\leq\Lambda\vert\gamma_n(x)-\gamma_n(y)\vert,\qquad\forall\,x,\,y\in\Omega
	\label{eq:gamma_invLip}
\end{equation}
for a certain constant $\Lambda\geq 0$. In this case, a suitable constant is $C=\Lambda^d$, indeed Eq. \eqref{eq:gamma_invLip} entails on the one hand $\Lebesgue^d(\gamma_n^{-1}(E))\leq\Lambda^d\Lebesgue^d(\gamma_n(\gamma_n^{-1}(E)))$, while on the other hand it suffices to observe that $\gamma_n(\gamma_n^{-1}(E))\subseteq E$, thus $\Lebesgue^d(\gamma_n(\gamma_n^{-1}(E)))\leq\Lebesgue^d(E)$. We claim that the relation \eqref{eq:gamma_invLip} holds if, for instance, the velocity $v_n$ is Lipschitz continuous on $\Omega$ with Lipschitz constant $0\leq L<\Delta{t}^{-1}$:
$$ \vert v_n(x)-v_n(y)\vert\leq L\vert x-y\vert, \qquad \forall\,x,\,y\in \Omega. $$
In fact, in such a case we have
\begin{align*}
	\vert\gamma_n(x)-\gamma_n(y)\vert &= \vert(x-y)+\Delta{t}(v_n(x)-v_n(y))\vert \\
	&\geq \vert x-y\vert-\Delta{t}\vert v_n(x)-v_n(y)\vert \\
	&\geq (1-L\Delta{t})\vert x-y\vert,
\end{align*}
and if $L$ fulfills the previous constraints we can take $\Lambda={(1-L\Delta{t})}^{-1}\geq 1$.

Another possible motion mapping complying with property \eqref{ass:gamma-Cn} is
$$ \gamma_n(x)=x+\Delta{t}\sum_{i=1}^M a_i^n\chi_{E_i}(x), $$
where $\{E_i\}_{i=1}^M$ is a pairwise disjoint partition of $\Omega$, i.e., $\Int{E_i}\cap\Int{E_j}=\emptyset$ each $i\ne j$ and $\cup_{i=1}^M E_i=\Omega$, $\chi_{E_i}$ is the characteristic function of $E_i$, and $a_i^n\in\R^d$ are constant, so that the velocity $v_n$ turns out to be the piecewise constant function $v_n(x)=\sum_{i=1}^M a_i^n\chi_{E_i}(x)$. Such a $\gamma_n$ is a measurable piecewise translation on $\Omega$ with the following property: For any $E\in\B(\Omega)$, it results
$$ \Lebesgue^d(\gamma_n^{-1}(E))=\sum_{i=1}^M\Lebesgue^d(\gamma_n^{-1}(E)\cap E_i)=
	\sum_{i=1}^M\Lebesgue^d(E\cap\gamma_n(E_i))\leq M\Lebesgue^d(E), $$
due to the invariance of Lebesgue measure under translations. A possible constant for Eq. \eqref{eq:ass-gamma-Cn} is thus $C=M$, the number of elements of the partition, although it may not be the optimal (i.e., the smallest) one. For instance, in case of constant velocity $v_n(x)=a\in\R^d$ in $\Omega$, which gives a linear transport $\gamma_n(x)=x+a\Delta{t}\chi_\Omega(x)$, the sets $\gamma_n(E_i)$ are pairwise disjoint and the last sum in the computation above equals $\Lebesgue^d(E\cap\gamma_n(\Omega))\leq\Lebesgue^d(E)$, so that Eq. \eqref{eq:ass-gamma-Cn} holds now more precisely with $C=1$ regardless of $M$.

If the initial measure $\mu_0$ is absolutely continuous with respect to $\Lebesgue^d$, a natural question is whether the same is true for all other measures $\mu_n$ recursively generated by the push forward. Property \eqref{ass:gamma-Cn} of the motion mappings turns out to be designed precisely for this purpose, indeed we have:
\begin{theorem}	\label{theo:abscont}
If $\mu_0\ll\Lebesgue^d$ then $\mu_n\ll\Lebesgue^d$ for all $n>0$.
\end{theorem}
\begin{proof}
We proceed inductively on $n$. Assume $\mu_n\ll\Lebesgue^d$ for a certain $n$ and consider $E\in\B(\Omega)$ such that $\Lebesgue^d(E)=0$. Then $\Lebesgue^d(\gamma_n^{-1}(E))\leq C\Lebesgue^d(E)=0$, whence 
$$ \mu_{n+1}(E)=\mu_n(\gamma_n^{-1}(E))=0, $$
i.e., $\mu_{n+1}\ll\Lebesgue^d$. Owing to $\mu_0\ll\Lebesgue^d$, we finally get by induction the thesis.
\end{proof}

When a density $\rho_n\in L^1(\Omega)$, $\rho_n\geq 0$ a.e. in $\Omega$, exists for the measures $\mu_n$, Eq. \eqref{eq:pushfwd_4} rewrites as
\begin{equation}
	\lint_\Omega\eta(x)\rho_{n+1}(x)\,dx=\lint_\Omega\eta(\gamma_n(x))\rho_n(x)\,dx
	\label{eq:pushfwd_dens}
\end{equation}
for all bounded and Borel functions $\eta:\Omega\to\R$. The new unknowns of the problem are now the $\rho_n$'s, for which we can state the following properties.
\begin{theorem}	\label{theo:wellpos}
Let $\rho_0\in L^1(\Omega)$ be given, $\rho_0\geq 0$ a.e. in $\Omega$. Then there exists a unique sequence $\{\rho_n\}_{n>0}\subset L^1(\Omega)$, $\rho_n\geq 0$ a.e. in $\Omega$, solving Eq. \eqref{eq:pushfwd_dens} with
$\rho_0$ as initial datum, and moreover
$$ \|\rho_n\|_1=\|\rho_0\|_1, \qquad \forall\,n>0. $$
If, in addition, $\rho_0\in L^1(\Omega)\cap L^\infty(\Omega)$ then also $\rho_n\in L^1(\Omega)\cap L^\infty(\Omega)$ with
$$ \|\rho_n\|_\infty\leq C^n\|\rho_0\|_\infty, \qquad \forall\,n>0. $$
\end{theorem}
\begin{proof}
Existence of a sequence $\{\rho_n\}_{n>0}\subset L^1(\Omega)$, $\rho_n\geq 0$ a.e. in $\Omega$, solving Eq. \eqref{eq:pushfwd_dens} is implied by Theorem \ref{theo:abscont} if one understands $\rho_0$ as the density of
$\mu_0$ with respect to $\Lebesgue^d$.
\begin{myenumerate}
\item To obtain uniqueness, assume first that
$\rho_{n+1},\,\rho'_{n+1}\in L^1(\Omega)$ are such that
$$ \lint_\Omega\eta(x)\rho_{n+1}(x)\,dx=\lint_\Omega\eta(x)\rho'_{n+1}(x)\,dx=
	\lint_\Omega\eta(\gamma_n(x))\rho_n(x)\,dx, $$
that is
$$ \lint_\Omega\eta(x)(\rho_{n+1}(x)-\rho'_{n+1}(x))\,dx=0 $$
for all bounded and Borel functions $\eta:\Omega\to\R$. Confining the attention to the continuous $\eta$ compactly supported in $\Omega$, i.e., $\eta\in C_0(\Omega)$, we see that, owing to Riesz representation theorem, this relation defines the null functional in the dual space $(C_0(\Omega))'$. Therefore $(\rho_{n+1}-\rho'_{n+1})\,dx$ must be the null measure, which implies $\rho_{n+1}(x)=\rho'_{n+1}(x)$ for a.e. $x\in\Omega$, whence the $L^1$-uniqueness of the density at the $(n+1)$-th time step. Proceeding now inductively from $\rho_0$, we get the uniqueness in $L^1(\Omega)$ of the sequence of densities $\rho_n$.

\item When estimating the $L^1$-norm of the $\rho_n$'s we can take advantage of their nonnegativity to discover, from Eq. \eqref{eq:pushfwd_dens} with $\eta=\chi_\Omega$:
$$ \|\rho_{n+1}\|_1=\lint_\Omega\rho_n(x)\,dx=\|\rho_n\|_1. $$
Thus, by induction, $\|\rho_n\|_1=\|\rho_0\|_1$ each $n>0$ as desired.

\item In order to prove the last statement of the theorem, we consider a pairwise disjoint partition $\{E_i\}_{i=1}^M$ of $\Omega$ and construct the following step function:
$$ s(x)=\sum_{i=1}^M\alpha_i\chi_{E_i}(x), \qquad
	\alpha_i=\frac{1}{\Lebesgue^d(E_i)}\lint_{\gamma_n^{-1}(E_i)}\rho_n(x)\,dx $$
for every $E_i$ such that $\Lebesgue^d(E_i)\ne 0$. Notice that
\begin{equation*}
	\|s\|_1=\sum_{i=1}^M\alpha_i\Lebesgue^d(E_i)=
		\sum_{i=1}^M\lint_{\gamma_n^{-1}(E_i)}\rho_n(x)\,dx=
		\lint_{\gamma_n^{-1}(\Omega)}\rho_n(x)\,dx=\|\rho_n\|_1,
\end{equation*}
and moreover $\alpha_i\geq 0$, therefore $s\in L^1(\Omega)$, $s\geq 0$ a.e. in $\Omega$. If $\Lebesgue^d(E_i)=0$ for some $i$, then the corresponding coefficient $\alpha_i$ can be arbitrarily defined without affecting $s$ as an element of $L^1(\Omega)$.

\item Set
$$ m:=\max_{i=1,\,\dots,\,M}\Lebesgue^d(E_i). $$
We claim that, for $m\to 0^+$, $s$ converges pointwise almost everywhere in $\Omega$ to $\rho_{n+1}$. For this we observe, first of all, that, taking $\eta=\chi_{E_i}$ in Eq. \eqref{eq:pushfwd_dens}, the coefficients $\alpha_i$ can be rewritten as
$$ \alpha_i=\frac{1}{\Lebesgue^d(E_i)}\lint_{E_i}\rho_{n+1}(x)\,dx. $$
Then we fix $x\in\Omega$ and notice that there exists a unique index $1\leq i=i(x)\leq M$ such that $x\in E_{i(x)}$, because the partition is pairwise disjoint. Thus $s(x)=\alpha_{i(x)}$ and
\begin{align*}
	\left\vert\rho_{n+1}(x)-s(x)\right\vert &=
		\left\vert\rho_{n+1}(x)-\frac{1}{\Lebesgue^d(E_{i(x)})}\lint_{E_{i(x)}}\rho_{n+1}(y)\,dy\right\vert \\
	&\leq \frac{1}{\Lebesgue^d(E_{i(x)})}\lint_{E_{i(x)}}\left\vert\rho_{n+1}(x)-\rho_{n+1}(y)\right\vert\,dy.
\end{align*}
As $m\to 0^+$, $E_{i(x)}$ shrinks to $x$ and Lebesgue differentiation theorem implies that the right-hand side of the above inequality tends to zero, therefore
$$ \lim_{m\to 0^+} \vert\rho_{n+1}(x)-s(x)\vert=0 \qquad \text{for a.e.\ } x\in\Omega, $$
which gives the desired pointwise convergence.

\item Assume $\rho_n\in L^1(\Omega)\cap L^\infty(\Omega)$, then
$$ \alpha_i\leq\frac{\Lebesgue^d(\gamma_n^{-1}(E_i))}{\Lebesgue^d(E_i)}
	\|\rho_n\|_\infty\leq C\|\rho_n\|_\infty, \qquad \forall\,i=1,\,\dots,\,M, $$
hence $s\in L^\infty(\Omega)$ with $s(x)\leq C\|\rho_n\|_\infty$ for all $x\in\Omega$. Consequently, in view of the pointwise convergence, also $\rho_{n+1}(x)\leq C\|\rho_n\|_\infty$ for a.e. $x\in\Omega$, which says $\rho_{n+1}\in L^\infty(\Omega)$ with
$$ \|\rho_{n+1}\|_\infty\leq C\|\rho_n\|_\infty. $$
Proceeding inductively from $\rho_0$ yields finally $\rho_n\in L^\infty(\Omega)$ with the desired estimate on the $L^\infty$-norm. \qedhere
\end{myenumerate}
\end{proof}

\subsection{Spatial approximation}	\label{subsect:spat_approx}
We turn now our attention to the construction of a spatial approximation of the push forward \eqref{eq:pushfwd_dens}, which will result in a computational scheme for the numerical treatment of the problem
\begin{equation}
	\begin{cases}
		\text{Find\ } \{\rho_n\}_{n>0}\subset L^1(\Omega)\cap L^\infty(\Omega),\,
			\rho_n\geq 0\ \text{a.e. in\ } \Omega\ \text{for all\ } n>0,\ \text{such that} \\[0.3cm]
		\displaystyle{\lint_\Omega}\eta(x)\rho_{n+1}(x)\,dx=
			\displaystyle{\lint_\Omega}\eta(\gamma_n(x))\rho_n(x)\,dx, \qquad
			\forall\,\eta:\Omega\to\R\ \text{bounded and Borel}
	\end{cases}
	\label{eq:problem}
\end{equation}
for a prescribed initial datum $\rho_0\in L^1(\Omega)\cap L^\infty(\Omega)$, $\rho_0\geq 0$ a.e. in $\Omega$. Theorem \ref{theo:wellpos} gives the well posedness of this problem, i.e., existence and uniqueness of the solution, with suitable estimates on the $L^1$ and $L^\infty$-norms in terms of the corresponding norms of the initial datum.

To be definite, we consider as domain $\Omega$ a (hyper)cube in $\R^d$, that we partition by a family of nested pairwise disjoint grids $\{E_j^h\}_{j=1}^{M_h}$ made in turn of (hyper)cubes with edges of constant length $h>0$, so that $\Lebesgue^d(E_j^h)=h^d$ for each $j$:
$$ \bigcup_{j=1}^{M_h}E_j^h=\Omega, \qquad\Int(E_i^h)\cap\Int(E_j^h)=\emptyset\quad\forall\,i\ne j. $$
The nesting of the grids implies that, given $0<h''<h'$, for all $0\leq j\leq M_{h''}$ there exists $0\leq i\leq M_{h'}$ such that $E_j^{h''}\subseteq E_i^{h'}$. Essentially the same construction described below may be in principle repeated for domains and/or grid elements of different shapes, up to possible technicalities in the practical implementation of the resulting numerical scheme.

One of the most natural ways to approximate Problem \eqref{eq:problem} is to introduce the following piecewise constant functions over the grid $\{E_j^h\}_{j=1}^{M_h}$:
$$ P^n_h(x)=\sum_{j=1}^{M_h}\rho_{j,h}^n\chi_{E_j^h}(x), \qquad n\geq 0, $$
such that $P^n_h(x)\equiv\rho_{j,h}^n$ for all $x\in E_j^h$, and to replace $\rho_n$, $\rho_{n+1}$ in Eq. \eqref{eq:pushfwd_dens} by $P^n_h$, $P^{n+1}_h$ respectively. Notice that each $P^n_h$ is integrable over $\Omega$, with
$$ \|P^n_h\|_1=h^d\sum_{j=1}^{M_h}\vert\rho_{j,h}^n\vert. $$

In addition, one can in general expect that the true function $\gamma_n$, which operates the push forward from the time step $n$ to the time step $n+1$, is not known exactly: Most frequently, an approximation $g^n_h:\Omega\to\Omega$ of it is available, resulting from some discretization of the velocity $v_n$. In the sequel, we will consider specifically
\begin{equation*}
	g^n_h(x)=x+u^n_h(x)\Delta{t},
	\label{eq:gn}
\end{equation*}
where $u^n_h$ is piecewise constant over the grid $\{E_j^h\}_{j=1}^{M_h}$:
\begin{equation*}
	u^n_h(x)=\sum_{j=1}^{M_h}u^n_{j,h}\chi_{E_j^h}(x), \qquad u^n_{j,h}\in\R^d.
	\label{eq:un}
\end{equation*}

If the velocity field $v_n$ is continuous in $\Omega$, so that its pointwise values make sense, then we define
$$ u^n_{j,h}=v_n(x_j), $$
$x_j$ being the center of $E_j^h$. Otherwise, assuming $v_n\in{(L^1(\Omega))}^d$, we take the approximate velocity field $u_h^n$ to be defined by
$$ u^n_{j,h}=\frac{1}{h^d}\lint_{E_j^h}v_n(x)\,dx, $$
which, owing to Lebesgue differentiation theorem, converges to $v_n$ in ${(L^1(\Omega))}^d$ as the grid is refined, i.e., when $h\to 0^+$. Notice that in both cases $u^n_h$ converges pointwise to $v_n$ for $h\to 0^+$ (up to possibly passing to subsequences and discarding Lebesgue negligible subsets in the second case).

Besides this approximation property, we further require:
\begin{assumption}	\label{ass:CFL}
Let $h,\,\Delta{t}>0$ be such that
\begin{equation}
	\Delta{t}\max_{j=1,\,\dots,\,M_h}\|u^n_{j,h}\|\leq h,
	\label{eq:CFL}
\end{equation}
where $\|\cdot\|$ is any norm in $\R^d$.
\end{assumption}
Equation \eqref{eq:CFL} reminds of the classical CFL condition arising for the stability of numerical schemes designed to approximate hyperbolic conservation laws. In the present context, it imposes a limitation on the maximum displacement of any grid cell $E_i^h$ produced by $g^n_h$, so that when computing
$$ \Lebesgue^d((g^n_h)^{-1}(E_i^h))=\sum_{j=1}^{M_h}\Lebesgue^d({(g_h^n)}^{-1}(E_i^h)\cap E_j^h)=
	\sum_{j=1}^{M_h}\Lebesgue^d(E_i^h\cap g_h^n(E_j^h)) $$
we see that there is a fixed maximum number, say $c$, of non-empty intersections $E_i^h\cap g^n_h(E_j^h)$ for $i,\,j$ running from $1$ to $M_h$, hence
$$ \Lebesgue^d((g^n_h)^{-1}(E_i^h))\leq c\Lebesgue^d(E_i^h)=ch^d. $$
In addition, $g^n_h$ being a piecewise translation on $\Omega$, the measure $g^n_h\#\Lebesgue^d$ is absolutely continuous with respect to $\Lebesgue^d$, and its density $r^n_h\in L^1(\Omega)$ satisfies, for a.e. $x\in\Omega$, $r^n_h(x)\geq 0$ and:
$$ r^n_h(x)=\lim_{R\to 0}\frac{\Lebesgue^d((g^n_h)^{-1}(B_R(x)))}{\Lebesgue^d(B_R(x))}=
	\lim_{R\to 0}\frac{\sum_{j=1}^{M_h}\Lebesgue^d(B_R(x)\cap g^n_h(E_j^h))}{\Lebesgue^d(B_R(x))}. $$
For $R$ sufficiently small and for a.e. $x\in\Omega$ the ball $B_R(x)$ is fully contained into one of the grid cells, thus in view of Eq. \eqref{eq:CFL} we have
$$ r^n_h(x)\leq\lim_{R\to 0}\frac{c\Lebesgue^d(B_R(x))}{\Lebesgue^d(B_R(x))}=c $$
and finally
\begin{equation}
	\Lebesgue^d((g^n_h)^{-1}(E))=\lint_E r^n_h(x)\,dx\leq c\Lebesgue^d(E), \qquad \forall\,E\in\B(\Omega),
	\label{eq:gn-c}
\end{equation}
which compares with the analogous property of $\gamma_n$ stated by Assumption \ref{ass:gamma}.

Bearing all of this in mind, a numerical scheme is obtained by imposing that the $P^n_h$'s satisfy Eq. \eqref{eq:pushfwd_dens} with the choice $\eta=\chi_{E_i^h}$, $i=1,\,\dots,\,M_h$, under the action of the mapping $g^n_h$: This gives
$$ \lint_{E_i^h}P^{n+1}_h(x)\,dx=\lint_{(g^n_h)^{-1}(E_i^h)}P^n_h(x)\,dx, $$
whence, using the specific representation of $P^n_h$, $P^{n+1}_h$ , the following scheme is derived:
\begin{equation}
	\rho_{i,h}^{n+1}=\frac{1}{h^d}\sum_{j=1}^{M_h}\rho_{j,h}^n\Lebesgue^d({(g_h^n)}^{-1}(E_i^h)\cap E_j^h),
		\qquad i=1,\,\dots,\,M_h\ \text{and\ } n\geq 0.
	\label{eq:numscheme}
\end{equation}

It is customary for the sequel to define the measures $\{\lambda_h^n\}_{n\geq 0}$ such that
$$ d\lambda^n_h:=P^n_h\,dx. $$
Notice that Eq. \eqref{eq:numscheme} does \emph{not} correspond to the push forward of $\lambda^n_h$ by $g^n_h$, i.e., $\lambda^{n+1}_h\ne g^n_h\#\lambda^n_h$, indeed the measure $g^n_h\#\lambda^n_h$ is in general not piecewise constant on the grid $\{E_j^h\}_{j=1}^{M_h}$, even for piecewise constant $\lambda^n_h$ and for a piecewise translation $g^n_h$, while $\lambda^{n+1}_h$ is forced to be so by construction. However, the scheme \eqref{eq:numscheme} enjoys some nice properties that we briefly list below:
\begin{myenumerate}
\item It is \emph{positivity-preserving}, in the sense that, given a discretization $P^0_h$ of the initial density $\rho_0$ such that $\rho_{j,h}^0\geq 0$ for all $j$, it results $\rho_{j,h}^n\geq 0$ for all $j$ and all $n>0$. This is consistent with the expected nonnegativity of the true densities $\rho_n$'s.

\item It is \emph{conservative}, indeed
\begin{align*}
	\|P^{n+1}_h\|_1 &= h^d\sum_{i=1}^{M_h}\rho^{n+1}_{i,h}=\sum_{j=1}^{M_h}\rho^n_{j,h}\sum_{i=1}^{M_h}
		\Lebesgue^d({(g_h^n)}^{-1}(E_i^h)\cap E_j^h) \\
	&= \sum_{j=1}^{M_h}\rho^n_{j,h}\Lebesgue^d({(g_h^n)}^{-1}(\Omega)\cap E_j^h)=h^d\sum_{j=1}^{M_h}\rho^n_{j,h}=\|P^n_h\|_1,
\end{align*}
hence
$$ \|P^n_h\|_1=\|P^0_h\|_1, \qquad \forall\,n>0, $$
in accordance with the analogous property stated by Theorem \ref{theo:wellpos} for the $\rho_n$'s.

\item It is \emph{boundedness-preserving}, in the sense that if the initial datum $\rho_0$ is discretized in such a way that $P^0_h\in L^\infty(\Omega)$ uniformly in $h$, i.e., there exists a constant $B_0\geq 0$, independent of $h$, such that
$$ \|P^0_h\|_\infty=\max_{j=1,\,\dots,\,M_h}\vert\rho^0_{j,h}\vert\leq B_0, $$
then $P^n_h\in L^\infty(\Omega)$ uniformly in $h$ as well, with
$$ \|P^n_h\|_\infty\leq c^nB_0, \qquad \forall\,n>0. $$ 
This follows by induction from Eqs. \eqref{eq:numscheme} and \eqref{eq:gn-c}, whence we get
$$ \|P^{n+1}_h\|_\infty=\max_{i=1,\,\dots,\,M_h}\vert\rho_{i,h}^{n+1}\vert
	\leq\max_{i=1,\,\dots,\,M_h}\frac{\|P^n_h\|_\infty}{h^d}\Lebesgue^d({(g_h^n)}^{-1}(E_i^h))\leq c\|P^n_h\|_\infty, $$
and is again consistent with both the boundedness of the true densities $\rho_n$ claimed by Theorem \ref{theo:wellpos} and the related estimate on their $L^\infty$-norm.
\end{myenumerate}

Given any piecewise constant function $P_h\in L^1(\Omega)\cap L^\infty(\Omega)$ on the grid $\{E_j^h\}_{i=1}^{M_h}$:
$$ P_h(x)=\sum_{j=1}^{M_h}\rho_{j,h}\chi_{E_j^h}(x), \qquad \rho_{j,h}\geq 0\quad \forall\,j=1,\,\dots,\,M_h, $$
and defined the measure $d\lambda_h=P_h\,dx$, let us denote by $g_h^n\apppf\lambda_h$ the measure whose density with respect to $\Lebesgue^d$ is the function
$$ (g_h^n\apppf P_h)(x)=\sum_{i=1}^{M_h}\left(\frac{1}{h^d}\sum_{j=1}^{M_h}\rho_{j,h}
	\Lebesgue^d({(g_h^n)}^{-1}(E_i^h)\cap E_j^h)\right)\chi_{E_i^h}(x). $$
Notice that, using this notation, we can rewrite the scheme \eqref{eq:numscheme} compactly as $\lambda_h^{n+1}=g_h^n\apppf\lambda_h^n$.
From the previous reasonings we know that, in general, $(g_h^n\apppf\lambda_h)(E)\ne (g_h^n\#\lambda_h)(E)$ for $E\in \B(\Omega)$, however:
\begin{lemma}	\label{lemma:pushfwd}
For all grid elements $\{E_k^h\}_{k=1}^{M_h}$ it results
$$ (g_h^n\apppf \lambda_h)(E_k^h)=(g_h^n\# \lambda_h)(E_k^h). $$
\end{lemma}
\begin{proof}
\begin{myenumerate}
\item We preliminarily observe that $g_h^n\#\lambda_h\ll\Lebesgue^d$, indeed let $E\in\B(\Omega)$ be such that $\Lebesgue^d(E)=0$, then
\begin{align*}
	(g_h^n\#\lambda_h)(E) &= \lambda_h({(g_h^n)}^{-1}(E))=\lint_{{(g_h^n)}^{-1}(E)}P_h(x)\,dx \\
	&\leq \|P_h\|_\infty\Lebesgue^d({(g_h^n)}^{-1}(E))\leq\|P_h\|_\infty c\Lebesgue^d(E)=0.
\end{align*}
The density $g_h^n\# P_h\in L^1(\Omega)\cap L^\infty(\Omega)$ of $g_h^n\#\lambda_h$ with respect to $\Lebesgue^d$ is given, for a.e. $x\in\Omega$, by
\begin{align*}
	(g_h^n\# P_h)(x) &= \lim_{R\to 0}\frac{(g_h^n\#\lambda_h)(B_R(x))}{\Lebesgue^d(B_R(x))}
	=\sum_{j=1}^{M_h}\rho_{j,h}\lim_{R\to 0}\frac{1}{\Lebesgue^d(B_R(x))}\lint_{B_R(x)}\chi_{g^n_h(E_j^h)}(y)\,dy \\
	&= \sum_{j=1}^{M_h}\rho_{j,h}\chi_{g^n_h(E_j^h)}(x).
\end{align*}

\item Given any grid element $E_k^h$, $1\leq k\leq M_h$, we compute:
\begin{align*}
	(g_h^n\apppf\lambda_h)(E_k^h) &= \lint_{E_k^h}(g_h^n\apppf P_h)(x)\,dx=
	\sum_{j=1}^{M_h}\rho_{j,h}\Lebesgue^d(E_k^h\cap g_h^n(E_j^h)) \\
	&= \lint_{E_k^h}(g_h^n\# P_h)(x)\,dx=(g_h^n\#\lambda_h)(E_k^h)
\end{align*}
and we have the thesis. \qedhere
\end{myenumerate}
\end{proof}

From this result we can derive some stability properties of the scheme.
\begin{theorem}[One-step stability]	\label{theo:onestepstab}
Let $\mu:\B(\Omega)\to\R_+$ be a measure supported in $\Omega$, with density $\rho\in L^1(\Omega)\cap L^\infty(\Omega)$ with respect to $\Lebesgue^d$. Assume moreover that $\gamma_n:\Omega\to\Omega$ is a diffeomorphism. Then, for $h>0$ sufficiently small, there exist two constants $A,\,B\geq 0$ such that
$$ \sum_{j=1}^{M_h}\left\vert(\gamma_n\#\mu)(E_j^h)-(g_h^n\apppf\lambda_h)(E_j^h)\right\vert\leq
	A\|\rho-P_h\|_1+Bh. $$
\end{theorem}
\begin{proof}
\begin{myenumerate}
\item We begin by observing that
\begin{align*}
	\left\vert(\gamma_n\#\mu)(E_j^h)-(g_h^n\apppf\lambda_h)(E_j^h)\right\vert &\leq
	\left\vert(\gamma_n\#\mu)(E_j^h)-(\gamma_n\#\lambda_h)(E_j^h)\right\vert \\
	&\phantom{\leq}+\left\vert(\gamma_n\#\lambda_h)(E_j^h)-(g_h^n\apppf\lambda_h)(E_j^h)\right\vert,
\intertext{and further that, owing to Lemma \ref{lemma:pushfwd},}
	&= \left\vert(\gamma_n\#\mu)(E_j^h)-(\gamma_n\#\lambda_h)(E_j^h)\right\vert \\
	&\phantom{\leq}+\left\vert(\gamma_n\#\lambda_h)(E_j^h)-(g_h^n\#\lambda_h)(E_j^h)\right\vert.
\end{align*}
Set
$$ \I_{1,j}:=\left\vert(\gamma_n\#\mu)(E_j^h)-(\gamma_n\#\lambda_h)(E_j^h)\right\vert, \quad
	\I_{2,j}:=\left\vert(\gamma_n\#\lambda_h)(E_j^h)-(g_h^n\#\lambda_h)(E_j^h)\right\vert. $$
	
\item It is an easy consequence of Theorem \ref{theo:abscont} that $\gamma_n\#\mu$ is absolutely continuous with respect to $\Lebesgue^d$. Let $\gamma_n\#\rho\in L^1(\Omega)\cap L^\infty(\Omega)$ be its density, then
$$ \lint_E(\gamma_n\#\rho)(x)\,dx=\lint_{\gamma_n^{-1}(E)}\rho(x)\,dx=\lint_E\rho(\gamma_n^{-1}(y))
	\vert\det{(\gamma_n^{-1})'(y)}\vert\,dy, \qquad \forall\,E\in\B(\Omega), $$
where $(\gamma_n^{-1})'$ is the Jacobian matrix of $\gamma_n^{-1}$. From the arbitrariness of $E$ we get then
$$ (\gamma_n\#\rho)(x)=\rho(\gamma_n^{-1}(x))\vert\det{(\gamma_n^{-1})'(x)}\vert, \quad \text{for a.e.\ }x\in\Omega. $$	

\item We estimate now $\I_{1,j}$. Since $\gamma_n$ is a diffeomorphism, there exist two constants $\Lambda_+,\,\Lambda_->0$ such that
$$ \|(\gamma_n)'(x)\|\leq\Lambda_+, \quad  \|(\gamma_n^{-1})'(x)\|\leq\Lambda_-,
	\quad \text{for all\ } x\in\Omega, $$
whence
\begin{align*}
	\I_{1,j} &\leq \lint_{E_j^h}\left\vert\rho(\gamma_n^{-1}(x))-P_h(\gamma_n^{-1}(x))\right\vert
		\left\vert\det{(\gamma_n^{-1})'}(x)\right\vert\,dx \\
	&\leq \Lambda_-^d\lint_{\gamma_n^{-1}(E_j^h)}\vert\rho(x)-P_h(x)\vert
		\left\vert\det{\gamma_n'}(x)\right\vert\,dx\leq
		\Lambda_-^d\Lambda_+^d\lint_{\gamma_n^{-1}(E_j^h)}\vert\rho(x)-P_h(x)\vert\,dx.
\end{align*}
and finally
$$ \sum_{j=1}^{M_h}I_{1,j}\leq A\|\rho-P_h\|_1 $$
for $A=\Lambda_-^d\Lambda_+^d$.

\item The estimate on $\I_{2,j}$ requires a bit trickier reasoning as it involves the comparison between the push forward operated by the true motion mapping $\gamma_n$ and the one operated by its approximation $g^n_h$. By a procedure analogous to that used in the proof of Lemma \ref{lemma:pushfwd} to compute the density of $g_h^n\#\lambda_h$, it is straightforward to see that the density of $\gamma_n\#\lambda_h$ can be also written as
$$ (\gamma_n\#P_h)(x)=\sum_{j=1}^{M_h}\rho^n_{j,h}\chi_{\gamma_n(E_j^h)}(x)\vert\det{(\gamma_n^{-1})'(x)}\vert, $$
whence
\begin{align*}
	\I_{2,j} &\leq \sum_{i=1}^{M_h}\rho_{i,h}\lint_{E_j^h}\left\vert\chi_{\gamma_n(E_i^h)}(x)
		\vert\det{(\gamma_n^{-1})'(x)}\vert-\chi_{g^n_h(E_i^h)}(x)\right\vert\,dx \\
	&\leq \|P_h\|_\infty\sum_{i=1}^{M_h}\Biggl[\Lambda_-^d\Lebesgue^d(\gamma_n(E_i^h)\cap E_j^h\setminus g^n_h(E_i^h)) \\
	&\phantom{\leq}+\lint_{\gamma_n(E_i^h)\cap g^n_h(E_i^h)\cap E_j^h}\vert\det{(\gamma_n^{-1})'(x)}-1\vert\,dx
		+\Lebesgue^d(g^n_h(E_i^h)\cap E_j^h\setminus\gamma_n(E_i^h))\Biggr].
\end{align*}
Recall the formula $\det{(\gamma_n^{-1})'(x)}={(\det{\gamma_n'(y)})}^{-1}$ for $x=\gamma_n(y)$ and denote $I\in\R^{d\times d}$ the identity matrix. We have $\gamma_n'(y)=I+\Delta{t}v_n'(y)$, with furthermore $\|v_n'(y)\|\leq L$ all $y\in\Omega$ for a certain constant $L>0$ (due to the smoothness of $v_n$), whence $\vert\tr{v_n'(y)}\vert\leq L^d$ ($\tr{v_n'(y)}$ being the trace of $v_n'(y)$) all $y\in\Omega$ and finally
$$ \det{\gamma_n'(y)}=1+\Delta{t}\tr{v_n'(y)}+o(\Delta{t}) \qquad (\Delta{t}\to 0). $$
Thus we conclude, for $\Delta{t}\to 0$, $\det{\gamma_n'(y)}\sim 1+K\Delta{t}$ for a constant $K\in\R$. It follows ${(\det{\gamma_n'(y)})}^{-1}\sim 1-K\Delta{t}$, therefore, for $\Delta{t}$ sufficiently small, $\vert{(\det{\gamma_n'(y)})}^{-1}-1\vert\leq\vert K\vert\Delta{t}$, which allows to specialize the previous estimate of $\I_{2,j}$ as
\begin{align*}
	\I_{2,j}\leq\|P_h\|_\infty C^\ast
		\sum_{i=1}^{M_h}&\left[\Lebesgue^d(\gamma_n(E_i^h)\triangle g^n_h(E_i^h)\cap E_j^h)\right. \\
	&+\left.\Delta{t}\Lebesgue^d(\gamma_n(E_i^h)\cap g^n_h(E_i^h)\cap E_j^h)\right],
\end{align*}
where $C^\ast=\max{\{\Lambda_-^d,\,\vert K\vert,\,1\}}$ and $\triangle$ denotes the symmetric difference of two sets. Summing over $j$ further yields
$$ \sum_{j=1}^{M_h}\I_{2,j}\leq\|P_h\|_\infty C^\ast
	\sum_{i=1}^{M_h}\left[\Lebesgue^d(\gamma_n(E_i^h)\triangle g^n_h(E_i^h))+
		\Delta{t}\Lebesgue^d(\gamma_n(E_i^h)\cap g^n_h(E_i^h))\right]. $$
It is now simple to obtain
$$ \Lebesgue^d(\gamma_n(E_i^h)\cap g^n_h(E_i^h))\leq\Lebesgue^d(g^n_h(E_i^h))=h^d, $$
whereas in order to estimate the Lebesgue measure of $\gamma_n(E_i^h)\triangle g^n_h(E_i^h)$ we rely on the following argument. Given $x\in E_i^h$, the maximum distance between the points $\gamma_n(x)$ and $g^n_h(x)$ is estimated as $\vert\gamma_n(x)-g^n_h(x)\vert=\Delta{t}\vert v_n(x)-v_n(x_j)\vert\leq L\Delta{t}h$, therefore, by taking $x\in\partial E_i^h$, we see that the set $\gamma_n(E_i^h)\setminus g^n_h(E_i^h)$ is fully contained into the gap between the (hyper)cube $g^n_h(E_i^h)$ with edge  size $h$ and the concentric (hyper)cube with edge size $h+2L\Delta{t}h$. This entails, for $\Delta{t}$ small (say $\Delta{t}<1$):
$$ \Lebesgue^d(\gamma_n(E_i^h)\setminus g^n_h(E_i^h))\leq (h+2L\Delta{t}h)^d-h^d=O(\Delta{t}h^d). $$
Conversely, the set $g^n_h(E_i^h)\setminus\gamma_n(E_i^h)$ is fully contained into the gap between the (hyper)cube $g^n_h(E_i^h)$ with edge size $h$ and the concentric (hyper)cube with edge size $h-2L\Delta{t}h$ (assuming $\Delta{t}<(2L)^{-1}$ for consistency). Consequently:
$$ \Lebesgue^d(g^n_h(E_j^h)\setminus\gamma_n(E_j^h))\leq h^d-(h-2L\Delta{t}h)^d=O(\Delta{t}h^d), $$
whence finally $\Lebesgue^d(\gamma_n(E_j^h)\triangle g^n_h(E_j^h))=O(\Delta{t}h^d)$, which allows to complete the estimate of $\sum_j\I_{2,j}$ as (consider that $M_hh^d=\Lebesgue^d(\Omega)$)
$$ \sum_{j=1}^{M_h}\I_{2,j}\leq\|P_h\|_\infty C^{\ast\ast}\Delta{t}, $$
$C^{\ast\ast}$ being a new cumulative constant independent of either $h$, $\Delta{t}$ or $n$. Considering that condition \eqref{eq:CFL} forces the order of magnitude of the time step $\Delta{t}$ to be dictated by the order of the grid size $h$, the previous relation implies finally $\sum_jI_{2,j}\leq Bh$ for a suitable constant $B\geq 0$.

\item Collecting the results so far obtained we ultimately have
$$ \sum_{j=1}^{M_h}\left\vert(\gamma_n\#\mu)(E_j^h)-(g_h^n\apppf\lambda_h)(E_j^h)\right\vert\leq
	\sum_{j=1}^{M_h}(\I_{1,j}+\I_{2,j})\leq A\|\rho-P_h\|_1+Bh $$
and the claim of the theorem follows. \qedhere
\end{myenumerate}
\end{proof}

Theorem \ref{theo:onestepstab} applied to the measures $\{\mu_n\}_{n\geq 0}$ and $\{\lambda_h^n\}_{n\geq 0}$ gives the stability of the numerical scheme \eqref{eq:numscheme} with respect to the push forward \eqref{eq:pushfwd} at one time step:
$$ \sum_{j=1}^{M_h}\left\vert\mu_{n+1}(E_j^h)-\lambda_h^{n+1}(E_j^h)\right\vert\leq
	A\|\rho_n-P_h^n\|_1+Bh. $$
This amounts to controlling the total variation of the error $\mu_{n+1}-\lambda_h^{n+1}$ over the grid cells (left-hand side) by the $L^1$-distance of the densities of the same measures at the previous time step and the size of the grid (right-hand size). By slightly modifying the argument used in this proof we can also obtain a stability over several time steps.

\begin{theorem}[Multistep stability]	\label{theo:multistepstab}
Let $\gamma_n$ be like in Theorem \ref{theo:onestepstab} and assume that $v_n$ is uniformly bounded in time over $\Omega$. Then, for $h>0$ sufficiently small, there exist a constant $A_n\geq 0$ such that
$$ \max_{j=1,\,\dots,\,M_h}\left\vert\mu_n(E_j^h)-\lambda_h^n(E_j^h)\right\vert\leq
	\max_{j=1,\,\dots,\,M_h}\left\vert\mu_0(E_j^h)-\lambda_h^0(E_j^h)\right\vert+A_nh^d. $$
\end{theorem}
\begin{proof}
We proceed inductively on $n$ and consider
\begin{align*}
	\left\vert\mu_{n+1}(E_j^h)-\lambda_h^{n+1}(E_j^h)\right\vert &\leq\left\vert(\gamma_n\#\mu_n)(E_j^h)
		-(\gamma_n\#\lambda_h^n)(E_j^h)\right\vert \\
	&\phantom{\leq}+\left\vert(\gamma_n\#\lambda_h^n)(E_j^h)-(g_h^n\#\lambda_h^n)(E_j^h)\right\vert.
\end{align*}
\begin{myenumerate}
\item The second term at the right-hand side is estimated like $\I_{2,j}$ in the proof of Theorem \ref{theo:onestepstab}. In particular, owing to condition \eqref{eq:CFL} along with the uniform boundedness of the velocity $v_n$ on $\Omega$, only a finite constant number of intersections $\gamma_n(E_i^h)\triangle g_h^n(E_i^h)\cap E_j^h$, $\gamma_n(E_i^h)\cap g_h^n(E_i^h)\cap E_j^h$ are nonempty when $i$ runs from $1$ to $M_h$, whence, for $h$ sufficiently small,
$$ \left\vert(\gamma_n\#\lambda_h^n)(E_j^h)-(g_h^n\#\lambda_h^n)(E_j^h)\right\vert\leq
	\|P_h^n\|_\infty C^\ast h^{d+1}, $$
where $C^\ast$ is a constant independent of both $h$ and $n$, and $\|P_h^n\|_\infty\leq c^nB_0$.

\item The first term at the right-hand side coincides with $\I_{1,j}$ in the proof of Theorem \ref{theo:onestepstab}, but we need here a different estimate for it. Specifically, we write
\begin{align*}
	\left\vert(\gamma_n\#\mu_n)(E_j^h)-(\gamma_n\#\lambda_h^n)(E_j^h)\right\vert &\leq
		\left\vert\mu_n(E_j^h)-\lambda_h^n(E_j^h)\right\vert \\
	&\phantom{\leq}+\lint_{\gamma_n^{-1}(E_j^h)\triangle E_j^h}\vert\rho_n(x)-P_h^n(x)\vert\,dx,
\end{align*}
then estimate
\begin{align*}
	\lint_{\gamma_n^{-1}(E_j^h)\triangle E_j^h}\vert\rho_n(x)-P_h^n(x)\vert\,dx &\leq
		(C^n\|\rho_0\|_\infty+c^nB_0)\Lebesgue^d(\gamma_n^{-1}(E_j^h)\triangle E_j^h).
\end{align*}
The maximum distance between a point $x\in\Omega$ and its image $\gamma_n(x)$ is
$$ \vert x-\gamma_n(x)\vert\leq\Delta{t}\vert v_n(x)\vert\leq \Delta{t}V, $$
where $V\geq\|v_n\|_\infty$ comes for the uniform in time boundedness of the velocity $v_n$ on $\Omega$. Therefore, if $\Delta{t}$ is chosen in such a way that $\Delta{t}V\leq h$ then condition \eqref{eq:CFL} is fulfilled and in addition, thinking now of $x\in\partial\gamma_n^{-1}(E_j^h)$, the inverse image $\gamma_n^{-1}(E_j^h)$ is fully contained into the (hyper)cube centered at $x_j$ (the center of $E_j^h)$ with edge size $h+2\Delta{t}V$, whence
$$ \Lebesgue^d(\gamma_n^{-1}(E_j^h)\setminus E_j^h)\leq (h+2\Delta{t}V)^d-h^d=O(h^d). $$
On the other hand, we have
$$ \Lebesgue^d(E_j^h\setminus\gamma_n^{-1}(E_j^h))\leq\Lebesgue^d(E_j^h)=h^d $$
so that finally $\Lebesgue^d(\gamma_n^{-1}(E_j^h)\triangle E_j^h)=O(h^d)$ and for $h$ small this term dominates over the previous one.

\item Owing to the above estimates, we can find a constant $a_n>0$, which essentially depends on $n$ through the powers of $C$ and $c$, such that
$$ \left\vert\mu_{n+1}(E_j^h)-\lambda_h^{n+1}(E_j^h)\right\vert\leq
	\left\vert\mu_n(E_j^h)-\lambda_h^n(E_j^h)\right\vert+a_nh^d, $$
hence taking the maximum over $j=1,\,\dots,\,M_h$ and then summing telescopically over $n$ we get the thesis with $A_n=\sum_{i=0}^{n-1}a_i$. \qedhere
\end{myenumerate}
\end{proof}

The quantities $\left\vert\mu_n(E_j^h)-\lambda_h^n(E_j^h)\right\vert$, $j=1,\,\dots,\,M_h$, can be interpreted as the \emph{localization error} of the measure $\mu_n$ produced by the measure $\lambda_h^n$ over the grid cells; in other words, they say how much the cell-by-cell distribution of the measure $\lambda_h^n$ is a reliable approximation of the distribution of $\mu_n$. We incidentally notice that the cell-by-cell distribution of $\lambda_h^n$ is actually the most visible outcome when looking at the result of a numerical simulation. Theorem \ref{theo:onestepstab} guarantees stability, in a single time step, of the total variation of the localization error, but requires for this the stronger information on the $L^1$-distance between the exact and the approximate densities. Theorem \ref{theo:multistepstab} yields instead the uniform stability of the localization error over several time steps with respect to the grid size, starting from the initial error produced by the discretization of $\mu_0$.

\subsection{Nonlinear fluxes}	\label{subsect:nonlinflux}
In this subsection we briefly comment on how the previous theory applies when the velocity field $v_n$ depends explicitly on the measure $\mu_n$, $v_n=v_n[\mu_n]$, which gives rise to a nonlinear flux $\mu_nv_n[\mu_n]$. As a matter of fact, this structure turns out to be quite relevant for a wide range of applications, including pedestrian flow that we will address later, where the dynamics is strongly influenced by the current configuration of the system.

Given the density $\rho_0\in L^1(\Omega)$, $\rho_0\geq 0$ a.e. in $\Omega$, of $\mu_0$ with respect to $\Lebesgue^d$, let us focus on the existence of the densities $\{\rho_n\}_{n>0}$ in case of motion mappings $\gamma_n=\gamma_n[\mu_n]$ explicitly depending on the $\mu_n$'s. The first step consists in assuming that, for a certain $n>0$, the measure $\mu_n$ is absolutely continuous with respect to $\Lebesgue^d$, with density $\rho_n\in L^1(\Omega)$, $\rho_n\geq 0$ a.e. in $\Omega$. Next one checks whether, in view of this hypothesis, property \eqref{ass:gamma-Cn} of Assumption \ref{ass:gamma} is actually satisfied. If this is the case, then Theorem \ref{theo:abscont} implies $\mu_{n+1}\ll\Lebesgue^d$ with density $\rho_{n+1}\in L^1(\Omega)$, $\rho_{n+1}\geq 0$ a.e. in $\Omega$, thus one concludes by induction that $\mu_n\ll\Lebesgue^d$ for each $n\geq 0$ as desired. In practice, this is nothing but the same reasoning scheme applied in Subsect. \ref{subsect:pushfwd_timedisc} to prove Theorem \ref{theo:abscont}, with the only difference that property \eqref{ass:gamma-Cn} of Assumption \ref{ass:gamma} on $\gamma_n$ is not assumed to hold \emph{a priori}, but is checked \emph{a posteriori} as a consequence of the inductive hypothesis $\mu_n\ll\Lebesgue^d$. Notice that uniqueness of the $\rho_n$'s along with $L^1$ and possibly $L^\infty$ estimates now follow straightforwardly, because the proof of Theorem \ref{theo:wellpos} does not rely on any specific structure of the motion mapping $\gamma_n$, except that it satisfies Eq. \eqref{eq:ass-gamma-Cn}.

As far as the spatial approximation of the $\mu_n$'s is concerned, one has to take into account that the approximate motion mapping $g^n_h$ now depends on $\lambda^n_h$ because so does the approximate velocity field $u^n_h$. Basically, one tries to mimic the true velocity $v_n[\mu_n]$ by $u^n_h[\lambda^n_h]$, setting
$$ u^n_{j,h}[\lambda^n_h]=v_n[\lambda^n_h](x_j) \qquad\text{or}\qquad
	u^n_{j,h}[\lambda^n_h]=\frac{1}{h^d}\lint_{E_j^h}v_n[\lambda^n_h](x)\,dx $$
and then reconstructing
$$ u^n_h[\lambda^n_h](x)=\sum_{j=1}^{M_h}u^n_{j,h}[\lambda^n_h]\chi_{E_j^h}(x). $$

By inspecting the proof of Theorems \ref{theo:onestepstab}, \ref{theo:multistepstab} we see that a relevant part of the error estimate involving $\I_2$ relies on the Lipschitz estimate of the velocity $v_n$ within each grid cell $E_j^h$:
$$ \vert\gamma_n(x)-g^n_h(x)\vert\leq L\Delta{t}h, \qquad \forall\,x\in E_j^h. $$
When the motion mappings $\gamma_n$, $g^n_h$ depend on the measures $\mu_n$, $\lambda^n_h$, some considerations on the localization error produced by the scheme in a single time step might still be developed if the velocity $v_n$ is such that
\begin{equation}
	\vert v_n[\mu_n](x)-v_n[\lambda^n_h](x_j)\vert=O\left(h+\|\rho_n-P^n_h\|_{L^1(E_j^h)}\right),
		\qquad \forall\,x\in E_j^h.
	\label{eq:Lip_nonlinflux}
\end{equation}
We refrain from going into this issue, we simply point out that, as we will see more in detail in the application to pedestrian flow (cf. Sect. \ref{subsect:numappr}), an estimate of this form is more easily proved when $\gamma_n,\,g^n_h$ depend on $\mu_n,\,\lambda^n_h$ through $\rho_n,\,P^n_h$, respectively, in a \emph{non local} way.

\section{Macroscopic modeling of pedestrian flow}	\label{sect:modeling}
In this section we address the application of pedestrian flow using the time-evolving measure framework previously developed. Pedestrians in motion within a given walking area are an essentially discrete system, in which each person plays the role of an isolated agent moving in interaction with other surrounding agents, usually with the aim of reaching a particular destination. As such, this system should be conceptually described at the microscopic scale, i.e., by a system of ordinary differential equations tracing the evolution in time of the position of each single pedestrian. The possible coupling of these ODEs should reflect the interactions among pedestrians, namely the influence that each one of them has on the motion of the others. Some microscopic models of pedestrian flow are currently available in the literature, generally regarding pedestrians as rigid particles, whose motion is formally regulated by classical Newtonian laws of point mechanics:
$$ \dot{v}_i(t)=F_i(t), $$
where $v_i$ is the velocity of the $i$-th pedestrian, $F_i$ an overall force acting on her/him, while the index $i$ runs from $1$ to the total number of pedestrian considered in the model. However, the force field $F_i$ usually invokes nonclassical dynamic concepts, resorting mainly to the ideas of \emph{preferred direction of motion} and of \emph{comfort/discomfort} due to the distance from or the proximity to other pedestrians. Therefore, the $F_i$'s have not to be understood strictly as mechanical actions exerted by pedestrians on each other. For instance, Helbing and coworkers \cite{PhysRevE.51.4282,HeMoFaBo} introduce the concept of \emph{social} (or \emph{behavioral}) \emph{force}, which measures the internal motivation of the individuals in performing certain movements. Specifically, in their model pedestrians are regarded as points, and two main factors contribute to the definition of the social force $F_i$ acting on the $i$-th individual:
\begin{myenumerate}
\item A relaxation toward a \emph{desired velocity}, i.e., the velocity that the $i$-th individual should possess in order to reach her/his destination as comfortably as possible.
\item A \emph{repulsive effect} from neighboring pedestrians located too close, i.e., within the so-called \emph{private sphere} of the $i$-th pedestrian, or from walls and borders found in the walking area.
\end{myenumerate}
Essentially the same guidelines underlie the microscopic model of pedestrian flow by Maury and Venel \cite{MaVe}. The only difference is that pedestrians are now regarded as rigid disks, hence the repulsive effect among them is meant here as a geometrical constraint to avoid that disks step over one another. An alternative formulation, still relying on a microscopic description of the system, is instead used by Hoogendoorn and Bovy \cite{HoBo,HoBo2}, who propose a theory of pedestrian behavior based on the concepts of walking task and walking cost. Basically, they assume that pedestrians are feedback-oriented controllers, who plan their movements on the basis of some predictions they make on the behavior of the other individuals. Predictions are dictated by a sort of cooperative or non-cooperative game theory; in either case, they are affected by a limited in time and space predictive ability of the walkers. Each pedestrian behaves so as to minimize her/his individual estimated walking cost, which is expressed by a suitable functional depending on the predicted positions of the other people.

A different approach to the description of the system uses instead partial differential equations and the theory of (possibly multidimensional) conservation laws. In this case, it is assumed that pedestrians moving within a given walking area have a continuous distribution in space, so that it makes sense to introduce their density $\rho=\rho(t,\,x)$ and to invoke some conservation principles, e.g. the conservation of mass and possibly also of linear momentum, in order to get an equation, or a system of equations, satisfied by $\rho$. Colombo and Rosini \cite{MR2158218} propose a one-dimensional model built on a Cauchy problem for the nonlinear hyperbolic conservation law
$$ \partial_t\rho+\partial_xf(\rho)=0, \qquad x\in\R,\ t>0, $$
which reminds of the Lighthill-Whitham-Richards (LWR) model of vehicular traffic (see Lighthill and Whitham \cite{MR0072606}, Richards \cite{MR0075522}). The main difference is that the density of pedestrians exhibits two characteristic maximal values $R,\,R^\star$, with $0<R<R^\star$, at both of which the flux $f$ vanishes. In normal situations $\rho$ ranges in the interval $[0,\,R]$, where the flux is nonnegative, either strictly concave or with at most one inflection point, and has precisely one local maximum point. When the density grows above its ``standard'' maximum value, i.e., for $\rho\in (R,\,R^\star]$, it is assumed that, unlike vehicular traffic, pedestrians can still move but feel overcompressed, hence their flux is less effective than before and they enter a \emph{panic} state. In this region, the function $f$ features a trend similar to that described for $\rho\in[0,\,R]$, but with a local maximum value strictly less than the previous one. As illustrated in \cite{MR2158218,CoRo}, this allows to define a concept of solution to the above conservation law in which non-classical shocks are admitted, i.e., shocks complying with the Rankine-Hugoniot condition but possibly violating entropy criteria. As a consequence, the classical maximum principle for nonlinear hyperbolic equations, stating that the solution $\rho(t,\,x)$ remains confined within the same lower and upper bounds of the initial datum for all $x\in\R$ and all $t>0$, no longer holds true and the model is able to describe the transition of pedestrians to panic even starting from an initial density entirely bounded below the standard maximum $R$. The resulting fundamental diagram, i.e., the mapping $\rho\mapsto f(\rho)$, agrees well with experimental observations reported by Helbing \emph{et al.} in \cite{HeJoAA}.

Bellomo and Dogb\'e \cite{MR2438218} refer instead to a two-dimensional setting, in which the walking area is represented by a bounded domain $\Omega\subset\R^2$ with possible inlet and outlet regions along the boundary $\partial\Omega$. The motion of pedestrians is described by a system of two partial differential equations invoking the conservation of mass and the balance of linear momentum:
$$
\begin{cases}
	\partial_t\rho+\nabla\cdot(\rho v)=0 \\
	\partial_tv+(v\cdot\nabla)v=F[\rho,\,v],
\end{cases}
$$
where $F$ is a material model for the acceleration of the individuals depending in general in a functional way on the density $\rho$ and the velocity $v$. The above equations are formally inspired by the classical fluid dynamics models of continuum mechanics, however the force (per unit mass) $F$ contains non-classical contributions accounting for:
\begin{myenumerate}
\item A relaxation toward a desired velocity, that makes pedestrians point in the direction of a certain target they want to reach.
\item A local crowding estimate, based on the pointwise values of $\nabla{\rho}$ possibly taken along the direction of the desired velocity, which might induce pedestrians to deviate from their preferred path in order to avoid areas of high density.
\item A pressure-like term, possibly regarded as a material quantity as in the celebrated Aw-Rascle model of vehicular traffic \cite{MR1750085}, which models the reaction of pedestrians to the presence of other individuals in the surrounding environment.
\end{myenumerate}
Additional topics, like e.g. the existence of a limited visibility zone for each pedestrian when trying to evaluate the minimal crowding direction, are also discussed. In particular, special attention is paid to the characterization of the panic state and to the transition to it from regular conditions: The Authors of \cite{MR2438218} suggest that pedestrians entering a panic state tend to follow chaotically other individuals, dropping any specific target, and therefore are mostly attracted toward areas of high density rather than seeking the less congested paths.

As anticipated in the Introduction, the usually large amount of coupled ODEs to be handled simultaneously required by microscopic models is a drawback if not from the computational point of view, thanks to the increasing power of modern calculators, certainly for analytical purposes. Among others, we like to mention here the chance to recover a global overview of the system from the knowledge of its microscopic state, possibly in connection with control and optimization issues. Macroscopic models are more suited to this, but those currently available in the literature have to face several other difficulties due to their intrinsic hyperbolic nature. First of all, the most natural settings for pedestrian flow problems are two-dimensional: One-dimensional models are essentially explorative, but they are unlikely to provide effective mathematical tools to deal with real applications. However, it is well known that the theory of multi-dimensional systems of nonlinear hyperbolic equations is much more complicated under both the analytical and the numerical point of view, therefore a sound mathematical mastery of such models may hardly be achieved. Secondly, the imposition of boundary conditions may be tricky in hyperbolic models, because on the one hand one is forcedly driven by the characteristic velocities in defining the inflow and the outflow portion of the boundary, while on the other hand it must be guaranteed that pedestrians do not enter or exit the domain from any point of the boundary other than the prescribed inlet and outlet regions. This issue gets even more complicated in presence of \emph{obstacles}, which have to be understood as internal boundaries to the walking area. We notice, however, that the most interesting problems for applications generally do not concern pedestrian motion in free spaces, but precisely in areas scattered with obstacles (e.g., pillars, bottlenecks, narrow passages, see Helbing \emph{et al.} \cite{HeMoFaBo}), sometimes used to force the flow of crowds in specific directions.

A time-evolving measures approach to the modeling of pedestrian flow may help overcome some of the difficulties just outlined. If one understands the mapping $\mu_n$ as a mathematical tool to evaluate the degree of \emph{space occupancy} by pedestrians, then it is perfectly natural to address the problem from a macroscopic, even Eulerian, point of view in spite of the intrinsic Lagrangian granularity of the system: Given a certain walking area $\Omega\subseteq\R^2$, the number $\mu_n(E)$ is simply a \emph{measure} of the crowding of the subset $E\subseteq\Omega$ at time $n$. In addition, as shown by the discussion of Sect. \ref{sect:theory}, there are basically no differences in the one- or two-dimensional theory, therefore one can immediately tackle realistic problems without the need for conceiving preliminary one-dimensional approximations of them. Finally, the imposition of at least the most common no-flux boundary condition, including hence the treatment of internal obstacles to the walking area, is quite straightforward, for in principle it simply requires to guarantee that $\supp{\mu_n}\subseteq\Omega$ for all $n$. Notice that if $O\subseteq\Omega$ is, e.g., an obstacle, then it might be sufficient to impose $\mu_n(O)=0$, which translates the idea that pedestrians cannot occupy the area covered by $O$. Then any other subset $E$ of $\Omega$ is automatically measured by taking the presence of the obstacle into account, indeed, assuming conveniently that $\mu_n$ is complete, it results $\mu_n(E\cap O)=0$, thus $\mu_n(E)=\mu_n(E\setminus O)$.

\subsection{Nonlocal flux}
The mathematical structure depicted by Eqs. \eqref{eq:pushfwd}, \eqref{eq:gamman} requires as major modeling task to devise the velocity $v_n$ at the $n$-th time step as a function of $x\in\Omega$ and possibly also of the current distribution $\mu_n$ of pedestrians. Following some ideas proposed by Maury and Venel \cite{MaVe} in the frame of microscopic models of crowd motion, we distinguish two main contributions to the overall velocity of pedestrians:
$$ v_n[\mu_n](x)=v_d(x)+\nu_n[\mu_n](x). $$

The \emph{desired velocity} $v_d=v_d(x)$ is the velocity a pedestrian would have in the absence of other surrounding pedestrians. This component of the total velocity describes the preferred direction of motion toward specific targets, possibly taking into account the presence of intermediate obstacles within the domain $\Omega$. Therefore, it is not affected by the actual crowding of the environments, but it specifically depends on the geometry of the walking area (in this sense, it is a sort of \emph{field velocity}).

\begin{figure}[t]
\centering
\includegraphics{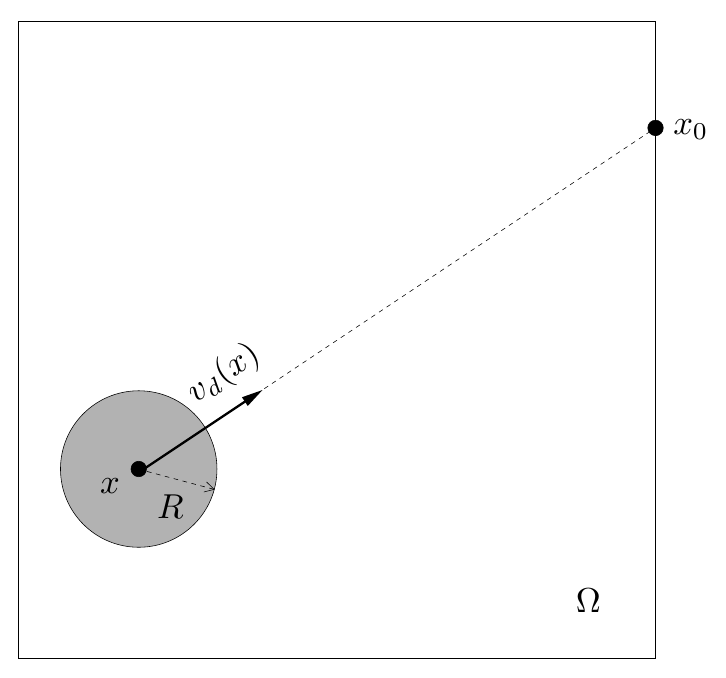}
\caption{The desired velocity $v_d$ toward the target $x_0$ and the interaction neighborhood $B_R(x)$ for a pedestrian located in $x\in\Omega$.}
\label{fig:domain}
\end{figure}

Assume the target of pedestrians is a certain location $x_0\in\bar\Omega$ representing, e.g., an aggregation point inside the walking area or a door along its boundary (as illustrated in Fig. \ref{fig:domain}). The simplest case is when $\Omega$ is star-shaped with respect to $x_0$, so that from any point $x\in\Omega$ there exists a straight path to $x_0$. Then the desired velocity at $x$ is readily defined as a vector pointing to $x_0$, namely
\begin{equation}
	v_d(x)=p_d(x)\frac{x_0-x}{\vert x_0-x\vert},
	\label{eq:vd}
\end{equation}
where $p_d:\Omega\to\R_+$ is a scalar nonnegative function representing the magnitude of $v_d$ (in other words, the \emph{desired speed} of pedestrians).

When the domain $\Omega$ is scattered with obstacles, one has to take into account that some points $x\in\Omega$ cannot be directly connected to $x_0$ by straight paths. For such points Eq. \eqref{eq:vd} must be duly modified, for instance by identifying several intermediate targets to be preliminarily reached, so as to bypass obstacles before pointing to $x_0$ (see Fig. \ref{fig:obstacle}). The definition of these intermediate targets is clearly not unique, and might be related to some strategy of walking optimization, for instance the minimization of the total length of the path from $x$ to $x_0$. The interested reader is referred to Maury and Venel \cite{MaVe} for further details. In this early application to pedestrian flow we refrain from going too much into this issue, deferring to a forthcoming work a deeper investigation of this modeling aspect.

\begin{figure}[t]
\centering
\includegraphics{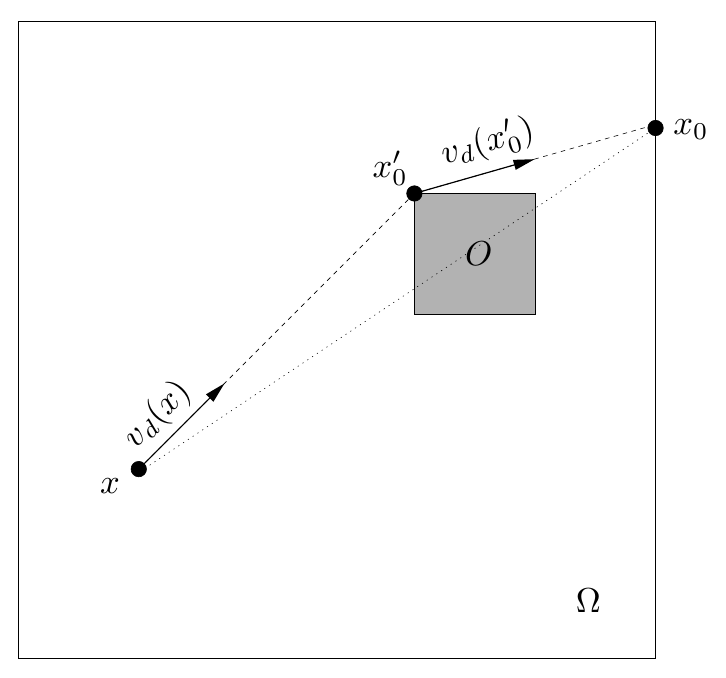}
\caption{A point $x\in\Omega$ that cannot be directly linked to the target $x_0\in\partial\Omega$ due to the presence of an intermediate obstacle $O\subset\Omega$. In this case, an intermediate target $x_0'$, coinciding with a corner of $O$, is identified, whence the actual target is then accessible. The desired velocity at $x$ points to $x_0'$, while the desired velocity at $x_0'$ finally points to $x_0$.}
\label{fig:obstacle}
\end{figure}

The \emph{interaction velocity} $\nu_n=\nu_n[\mu_n](x)$ expresses the deviation of pedestrians from their preferred path due to the presence of other surrounding pedestrians. In particular, we assume that interactions among pedestrians are effective only within a finite distance $R>0$, defining a \emph{neighborhood of interaction} $B_R(x)$ for each point $x\in\Omega$:
$$ B_R(x)=\{y\in\Omega\,:\,\|y-x\|\leq R\}. $$
The specific shape of this neighborhood depends on the $\R^2$-metric used in the above definition. If, for instance, $\|\cdot\|$ is the Euclidean norm then $B_R(x)$ is the closed ball centered at $x$ with radius $R$ (see Fig. \ref{fig:domain}). Instead, by taking $\|\cdot\|$ to be the $\infty$-norm the neighborhood of interaction results in a square centered at $x$ with side $2R$.

Resorting to the idea of nonlocal interactions among agents proposed by Canuto \emph{et al.} \cite{CaFaTi} for \emph{rendez-vous} problems, we assume that pedestrians interact in a nonlocal way with all subjects comprised within their neighborhood of interaction. The main difference with respect to the above-cited \emph{rendez-vous} is that in this case, at least under normal, i.e., no panic conditions, instead of aiming at aggregating they carefully try to bypass highly congested surrounding areas. We propose two different ways in which this behavior can be modeled, both based on the detection, by a weighted average, of a suitable point $x^\star\in B_R(x)$. In the first case, $x^\star$ is the center of mass of pedestrians in the neighborhood $B_R(x)$:
$$ x^\star=\frac{1}{\mu_n(B_R(x))}\lint_\Omega y\chi_{B_R(x)}(y)\,d\mu_n(y), $$
which can be thought of as an indicator of the area of highest crowd concentration within $B_R(x)$. Pedestrians trying to steer clear of this zone are driven away from $x^\star$, hence their interaction velocity can be given the form:
\begin{align}
	\nu_n[\mu_n](x) &= p_\nu[\mu_n](x)(x-x^\star) \notag \\
	& = p_\nu[\mu_n](x)\frac{1}{\mu_n(B_R(x))}\lint_\Omega(x-y)\chi_{B_R(x)}(y)\,d\mu_n(y),
	\label{eq:nu_G}
\end{align}
$p_\nu[\mu_n]:\Omega\to\R_+$ being a nonnegative function related to the intensity of the interaction. The superposition of $v_d(x)$ and $\nu_n[\mu_n](x)$ may finally lead pedestrians to deviate from their preferred path in order to avoid local jams (see Fig. \ref{fig:superpos_vel}, left).

In the second case, $x^\star$ is instead a point of low crowding within the neighborhood of interaction, which can be computed as a sort of ``inverse center of mass''. In practice, it should average out the neighboring locations $y\in B_R(x)$, emphasizing those with low pedestrian density $\rho_n(y)$:
$$ x^\star=\frac{\lint_\Omega y\chi_{B_R(x)}(y)(\omega\circ\rho_n)(y)\,dy}
	{\lint_\Omega\chi_{B_R(x)}(y)(\omega\circ\rho_n)(y)\,dy}, $$
where $\omega:\R_+\to\R_+$ is a nonnegative nonincreasing weight function. Once this point has been defined, pedestrians locally deviate toward it, hence we put
\begin{align}
	\nu_n[\mu_n](x) &= p_\nu[\mu_n](x)(x^\star-x)	\notag \\
	&= p_\nu[\mu_n](x)\frac{\lint_\Omega (y-x)\chi_{B_R(x)}(y)(\omega\circ\rho_n)(y)\,dy}
		{\lint_\Omega\chi_{B_R(x)}(y)(\omega\circ\rho_n)(y)\,dy}.
	\label{eq:nu}
\end{align}
Notice that, unlike Eq. \eqref{eq:nu_G}, in Eq. \eqref{eq:nu} the dependence of $\nu_n$ on $\mu_n$ is achieved through $\rho_n$, which explicitly requires that a density for the measure $\mu_n$ be defined with respect to $\Lebesgue^2$ for all $n>0$. Again, the superposition of $v_d(x)$ and $\nu_n[\mu_n](x)$ make pedestrians modify their direction of motion looking for uncrowded surrounding areas (see Fig. \ref{fig:superpos_vel}, right).

\begin{figure}[t]
\centering
\includegraphics{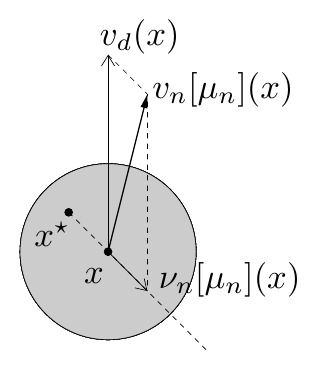} \hspace{3cm} \includegraphics{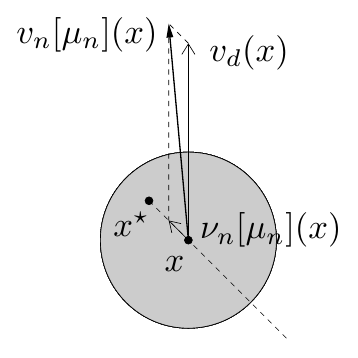}
\caption{Construction of the velocity $v_n[\mu_n](x)$ when $x^\star$ is the center of mass of pedestrians (left) or a point of low crowding (right) in $B_R(x)$.}
\label{fig:superpos_vel}
\end{figure}

In the sequel we will concentrate on the form of $\nu_n[\mu_n]$ provided by Eq. \eqref{eq:nu}, the case of Eq. \eqref{eq:nu_G} requiring, at least for $\mu_n\ll\Lebesgue^2$, essentially the same calculations up to the use of the boundedness of the density $\rho_n$ claimed by Theorem \ref{theo:wellpos}. In more detail, the case that we will be able to analyze theoretically is when
$$ p_d(x)=\alpha\vert x_0-x\vert,\qquad p_\nu[\mu_n](x)=\beta\lint_\Omega\chi_{B_R(x)}(y)(\omega\circ\rho_n)(y)\,dy, $$
for $\alpha,\,\beta\geq 0$, which gives rise to
\begin{equation}
	v_n[\mu_n](x)=\alpha(x_0-x)+\beta\lint_\Omega(y-x)\chi_{B_R(x)}(y)(\omega\circ\rho_n)(y)\,dy.
	\label{eq:vn}
\end{equation}
As a consequence, the motion mapping $\gamma_n$ depends on the measure $\mu_n$, therefore the theory presented in Sect. \ref{sect:theory} applies to the present model provided the discussion made in Subsect. \ref{subsect:nonlinflux} holds for the velocity $v_n[\mu_n]$ resulting from Eq. \eqref{eq:vn}.

Let us prove that $\gamma_n$ complies with property \eqref{ass:gamma-Cn} of Assumption \ref{ass:gamma} by showing that $v_n[\mu_n]$ is Lipschitz continuous in $\Omega$. This is self-evident for the desired velocity $v_d(x)=\alpha(x_0-x)$, while, as far as the interaction velocity $\nu_n[\mu_n]$ is concerned, we compute, for $x_1,\,x_2\in\Omega$,
\begin{align*}
	\vert\nu_n[\mu_n](x_2)-\nu_n[\mu_n](x_1)\vert &=
		\beta\left\vert\lint_\Omega y\left(\chi_{B_R(x_1)}(y)-\chi_{B_R(x_2)}(y)\right)(\omega\circ\rho_n)(y)\,dy
			\right. \\
	&\phantom{=}+\left.\lint_\Omega\left(x_2\chi_{B_R(x_2)}(y)-x_1\chi_{B_R(x_1)}(y)\right)(\omega\circ\rho_n)(y)\,dy
		\right\vert \\
	&\leq\beta\omega(0)\left[C_\Omega\lint_\Omega\vert\chi_{B_R(x_1)}(y)-\chi_{B_R(x_2)}(y)\vert\,dy\right. \\
		&\phantom{=}+\left.\lint_\Omega\vert x_2\chi_{B_R(x_2)}(y)-x_1\chi_{B_R(x_1)}(y)\vert\,dy\right],
\end{align*}
where $C_\Omega>0$ is a constant related to the boundedness of $\Omega$, such that $\vert x\vert\leq C_\Omega$ for all $x\in\Omega$, and where we have used $\omega(s)\leq\omega(0)$ for all $s\geq 0$, together with the inductive hypothesis $\rho_n\geq 0$ a.e. in $\Omega$, the weight $\omega$ being nonincreasing on $[0,\,+\infty)$. Notice that for any two sets $A,\,B$ and any two real numbers $a,\,b$ it results
$$	\vert a\chi_A(x)-b\chi_B(x)\vert=
		\begin{cases}
			\vert a\vert & \text{if\ } x\in A\setminus B \\
			\vert b\vert & \text{if\ } x\in B\setminus A \\
			\vert a-b\vert & \text{if\ } x\in A\cap B \\
			0 & \text{otherwise},
		\end{cases} $$
thus in particular, when $a=b=1$, $\vert\chi_A(x)-\chi_B(x)\vert=\chi_{A\triangle B}(x)$, where $A\triangle B=(A\setminus B)\cup (B\setminus A)$ is the symmetric difference of the sets $A$ and $B$. Therefore we can carry the previous computation on as
\begin{align*}
	\vert\nu_n[\mu_n](x_2)-\nu_n[\mu_n](x_1)\vert &\leq
		\beta\omega(0)[C_\Omega\Lebesgue^2(B_R(x_1)\triangle B_R(x_2))+\vert x_2\vert
		\Lebesgue^2(B_R(x_2)\setminus B_R(x_1)) \\
	&\phantom{\leq}+\vert x_1\vert\Lebesgue^2(B_R(x_1)\setminus B_R(x_2))+
		\vert x_2-x_1\vert\Lebesgue^2(B_R(x_1)\cap B_R(x_2))].
\end{align*}
Assuming that the neighborhoods of interaction are closed balls with respect to the Euclidean metric in $\R^2$ and invoking the invariance of the Lebesgue measure under translations, it is basically a plane geometry task to see that
$$ \Lebesgue^2(B_R(x_1)\triangle B_R(x_2))\ \text{is\ }
	\begin{cases}
		=2\pi R^2 & \text{if\ } \vert x_2-x_1\vert\geq 2R \\
		\leq 4R\vert x_2-x_1\vert & \text{if\ } \vert x_2-x_1\vert<2R,
	\end{cases} $$
thus globally
\begin{equation}
	\Lebesgue^2(B_R(x_1)\triangle B_R(x_2))\leq 4R\vert x_2-x_1\vert.
	\label{eq:L2Lip}
\end{equation}
From this we deduce immediately
$$ \Lebesgue^2(B_R(x_1)\setminus B_R(x_2))=\Lebesgue^2(B_R(x_2)\setminus B_R(x_1))=\frac{1}{2}\Lebesgue^2(B_R(x_1)\triangle B_R(x_2))
	\leq 2R\vert x_2-x_1\vert, $$
whence finally, using further $\Lebesgue^2(B_R(x_1)\cap B_R(x_2))\leq\pi R^2$,
$$ \vert\nu_n[\mu_n](x_2)-\nu_n[\mu_n](x_1)\vert\leq\beta\omega(0)R(8C_\Omega+\pi R)\vert x_2-x_1\vert $$
which yields the Lipschitz continuity of the interaction velocity of pedestrians. Analogous calculations can be repeated for square-shaped neighborhoods of interaction.

Owing to the above results, we conclude that $v_n[\mu_n]$ is Lipschitz continuous in $\Omega$ as it is the algebraic sum of two Lipschitz continuous functions. The reasonings of Subsect. \ref{subsect:pushfwd_timedisc} entail then that $\gamma_n$ complies with Eq. \eqref{eq:ass-gamma-Cn}, hence in view of Theorems \ref{theo:abscont}, \ref{theo:wellpos} we have proved:
\begin{proposition}	\label{prop:mod-abscont}
Let $\mu_0\ll\Lebesgue^2$ with density $\rho_0\in L^1(\Omega)\cap L^\infty(\Omega)$, $\rho_0\geq 0$ a.e. in $\Omega$, and let the velocity field \eqref{eq:vn} be given. Then:
\begin{myenumerate} 
\item $\mu_n\ll\Lebesgue^2$ for all $n>0$.
\item There exists a unique sequence $\{\rho_n\}_{n>0}\subset L^1(\Omega)\cap L^\infty(\Omega)$, $\rho_n\geq 0$ a.e. in $\Omega$, such that $d\mu_n=\rho_n\,dx$ each $n>0$, solving Eq. \eqref{eq:pushfwd} (or, equivalently, Eq. \eqref{eq:pushfwd_dens}) with $\rho_0$ as initial datum.
\end{myenumerate}
\end{proposition}
We can therefore speak of \emph{density of pedestrians} (in the sense of Radon-Nikodym theorem), and recover the measure $\mu_n$ as
$$ \mu_n(E)=\lint_E\rho_n(x)\,dx, \qquad \forall\,E\in\B(\Omega) $$
each $n\geq 0$. Remember that the $\mu_n$-measure of the subsets of $\Omega$, more than the pointwise values of the density $\rho_n$, is in principle the actual meaningful macroscopic information describing the space occupancy by the crowd.

\begin{figure}[t]
\centering
\includegraphics{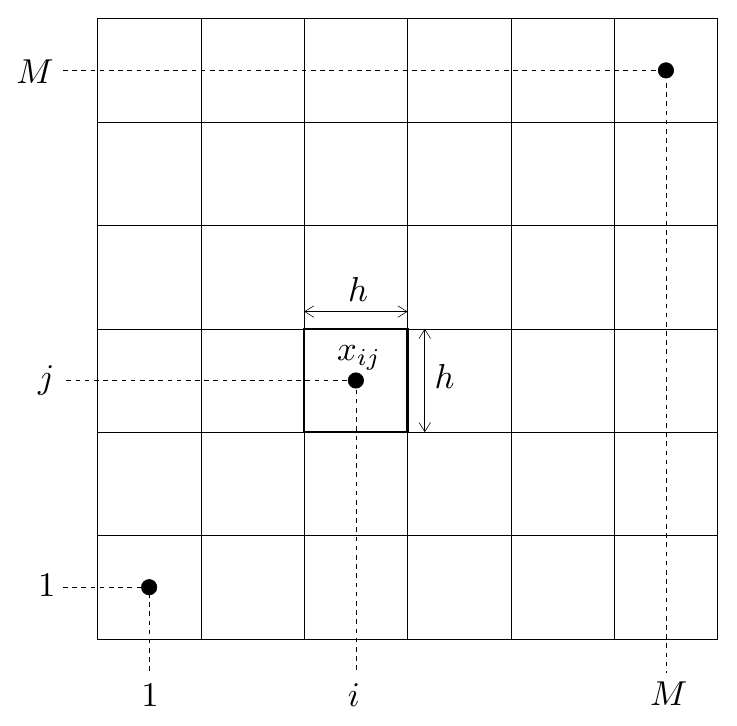}
\caption{The grid cell $E_{ij}$ with side $h$ and its center $x_{ij}$.}
\label{fig:grid}
\end{figure}

\subsection{Numerical approximation}	\label{subsect:numappr}
Spatial approximation of Problem \eqref{eq:problem} can be carried out by fixing preliminarily the walking area to be, for instance, the square $\Omega=[0,\,1]\times[0,\,1]\subset\R^2$. Let us partition $\Omega$ by a uniform orthogonal grid, made of square elements of side $h$, that we denote by $\{E_{ij}\}_{i,\,j=1}^M$ (from now on we drop the index $h$ for notational convenience). In this specific two-dimensional application, it is useful to identify each element of the grid by a couple of indices $i,\,j$, which locate 
the center $x_{ij}$ of the cell as
$$ x_{ij}=\left(\frac{2i-1}{2}h,\,\frac{2j-1}{2}h\right), \qquad i,\,j=1,\,\dots,\,M, $$
see Fig. \ref{fig:grid}, and to deal consequently with a double-indexed numerical density $P^n$:
$$ P^n(x)=\sum_{i,\,j=1}^M\rho^n_{ij}\chi_{E_{ij}}(x), \qquad n\geq 0. $$

In particular, the numerical scheme \eqref{eq:numscheme} rewrites as
\begin{equation}
	\rho^{n+1}_{ij}=\frac{1}{h^2}\sum_{l,\,m=1}^M\rho^n_{lm}\Lebesgue^2(E_{ij}\cap g^n[\lambda^n](E_{lm})), \qquad
		i,\,j=1,\,\dots,\,M,\ n\geq 0,
	\label{eq:numscheme_ped}
\end{equation}
where we have used the invariance of $\Lebesgue^2$ under the (piecewise) translation $g^n[\lambda^n]$ to get
$$ \Lebesgue^2((g^n[\lambda^n])^{-1}(E_{ij})\cap E_{lm})=\Lebesgue^2(E_{ij}\cap g^n[\lambda^n](E_{lm})). $$
The dependence of the approximate motion mapping $g^n$ on the measure $d\lambda^n=P^n\,dx$ amounts naturally to a dependence of the measures $\Lebesgue^2(E_{ij}\cap g^n[\lambda^n](E_{lm}))$ on the coefficients $\{\rho^n_{ij}\}_{i,\,j=1}^M$ of $P^n$. Notice however that the form \eqref{eq:numscheme_ped} of the scheme is more convenient for implementation purposes than the generic form \eqref{eq:numscheme}, as it does not require to determine the inverse images of the grid cells with respect to the mapping $g^n$. Rather, we incidentally notice that their images are straightforwardly obtained by translation as $g^n[\lambda^n](E_{lm})=E_{lm}+u^n_{lm}[\lambda^n]\Delta{t}$.

According to the discussion of Subsect. \ref{subsect:nonlinflux}, $g^n$ complies with property \eqref{ass:gn-CFL} of Assumption \ref{ass:CFL} provided at each time step the CFL condition \eqref{eq:CFL} is satisfied. In more detail, for every fixed pair of indices $(i,\,j)$ the intersection $E_{ij}\cap g^n[\lambda^n](E_{lm})$ is nonempty for at most nine adjacent pairs of indices $(l,\,m)$, namely
$$ (l,\,m)=(i,\,j),\ (i\pm 1,\,j),\ (i,\,j\pm 1),\ (i\pm 1,\,j\pm 1), $$
so that, denoting by $U_1,\,U_2$ the horizontal and vertical component, respectively, of $u^n_{lm}[\lambda^n]$, the coefficients $\Lebesgue^2(E_{ij}\cap g^n[\lambda^n](E_{lm}))$ of the above scheme can be duly computed as (cf. Fig. \ref{fig:intersect})
\begin{align*}
	\Lebesgue^2(E_{ij}\cap g^n[\lambda^n](E_{lm})) &= [U_1^+\Delta{t}\delta_{l,i-1}+(h-\vert U_1\vert\Delta{t})\delta_{li}
		+U_1^-\Delta{t}\delta_{l,i+1}] \\
	&\phantom{=}\times[U_2^+\Delta{t}\delta_{m,j-1}+(h-\vert U_2\vert\Delta{t})\delta_{mj}
		+U_2^-\Delta{t}\delta_{m,j+1}].
\end{align*}
In this formula, $\delta_{rs}$ is the Kronecker delta:
$$	\delta_{rs}=
		\begin{cases}
			1 & \text{if\ } r=s \\
			0 & \text{if\ } r\ne s,
		\end{cases} $$
whereas
$$ U^+=\max{(U,\,0)}, \qquad U^-=\max{(-U,\,0)} $$
are the positive and negative part, respectively, of $U$.

\begin{figure}[t]
\centering
\includegraphics{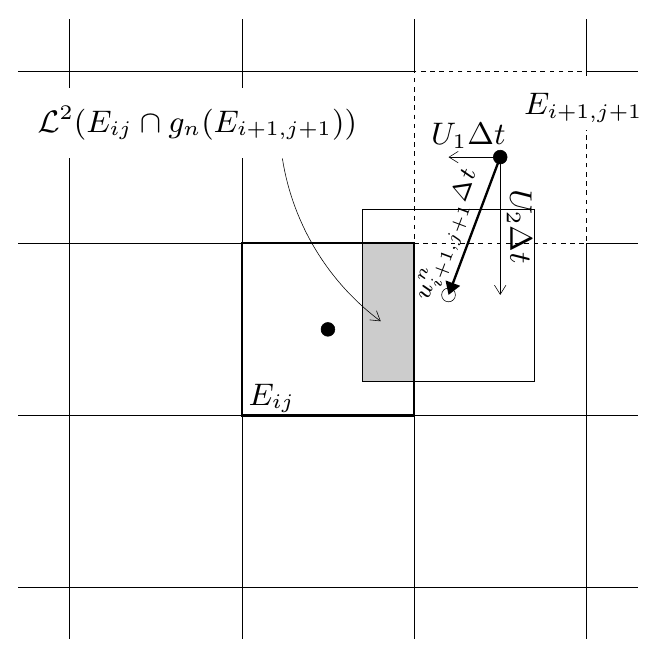}
\caption{The intersection between $E_{ij}$ and $g^n(E_{i+1,j+1})$, which is nonempty only if the horizontal and vertical components $U_1,\,U_2$ of the velocity $u^n_{i+1,j+1}$ are both negative. In this case, the Lebesgue measure (area) of the intersection is $\Lebesgue^2(E_{ij}\cap g^n(E_{i+1,j+1}))=U_1^-U_2^-\Delta{t}^2$.}
\label{fig:intersect}
\end{figure}

With respect to the discussion proposed in Subsect. \ref{subsect:nonlinflux}, let us compute, from the definition of $v_n$ given by Eq. \eqref{eq:vn}:
\begin{align*}
	v_n[\mu_n](x)-v_n[\lambda^n](\xi) &= \alpha(\xi-x)+\beta\lint_\Omega(y-x)\chi_{B_R(x)}(y)
		(\omega\circ\rho_n)(y)\,dy \\
	&\phantom{=}-\beta\lint_\Omega(y-\xi)\chi_{B_R(\xi)}(y)(\omega\circ P^n)(y)\,dy
\end{align*}
whence, adding and subtracting repeatedly several terms,
\begin{align*}
	\vert v_n[\mu_n](x)-v_n[\lambda^n](\xi)\vert &\leq \alpha\vert\xi-x\vert+
		\beta\left\vert\lint_\Omega(y-x)(\chi_{B_R(x)}(y)-\chi_{B_R(\xi)}(y))(\omega\circ\rho_n)(y)\,dy\right. \\
	&\phantom{\leq}+\lint_\Omega(y-x)\chi_{B_R(\xi)}(y)(\omega\circ\rho_n-\omega\circ P^n)(y)\,dy \\
	&\phantom{\leq}\left.+(\xi-x)\lint_\Omega\chi_{B_R(\xi)}(y)(\omega\circ P^n)(y)\,dy\right\vert \\
	&\leq \alpha\vert\xi-x\vert+2\beta C_\Omega[\omega(0)\Lebesgue^2(B_R(x)\triangle B_R(\xi))+L_\omega\|\rho_n-P^n\|_1] \\
	&\phantom{\leq}+\beta\omega(0)\Lebesgue^2(\Omega)\vert\xi-x\vert,
\intertext{where we have assumed that the weight function $\omega$ is Lipschitz continuous with Lipschitz constant $L_\omega>0$. Using Eq. \eqref{eq:L2Lip} we further obtain}
	&\leq [\alpha+8\beta C_\Omega\omega(0)+\beta\Lebesgue^2(\Omega)\omega(0)]\vert\xi-x\vert
		+2\beta C_\Omega L_\omega\|\rho_n-P^n\|_1 \\
	&=: C'\vert\xi-x\vert+C''\|\rho_n-P^n\|_1,
\end{align*}
which is precisely an estimate of the form \eqref{eq:Lip_nonlinflux}.

\section{Conclusions and research perspectives}	\label{sect:conclusions}
In this paper we have addressed the modeling of pedestrian flow problems from a macroscopic point of view, resorting to classical measure theoretical methods and specifically to the concept of discrete-time evolving measures.

The main idea consists in describing the space occupancy by pedestrians through a positive measure $\mu_n$, parameterized by a discrete time index $n\in\N$, which evolves according to the recursive push forward equation $\mu_{n+1}=\gamma_n\#\mu_n$. Here, $\gamma_n(x)=x+v_n(x)\Delta{t}$ is the so-called motion mapping, expressing the space displacement during a time step $\Delta{t}$ under the effect of a velocity field $v_n$. Given a walking area $\Omega\subseteq\R^2$, the number $\mu_n(E)$ represents, for every measurable $E\subseteq\Omega$, a measure, in macroscopic averaged terms, of the crowding of the area $E$ or, in other words, an estimate of the amount of people contained in $E$. Hence, the point of view on the system finally provided by the sequence of measures $\{\mu_n\}_{n>0}$ is essentially Eulerian, in spite of the Lagrangian description of the motion yielded by the mappings $\{\gamma_n\}_{n>0}$.

The first part of the work (cf. Sect. \ref{sect:theory}) has been devoted to a theoretical study, for generic $d$-dimensional systems, of the mathematical structures just outlined. In particular, we have first shown that the discrete-time push forward can be formally derived from an explicit time discretization of the classical mass conservation equation of continuum mechanics (cf. Proposition \ref{prop:mass-Lag-mu}). Then we have addressed specifically the question of the existence and uniqueness of a density $\rho_n$ for the measure $\mu_n$ with respect to the Lebesgue measure $\Lebesgue^d$ (cf. Theorems \ref{theo:abscont}, \ref{theo:wellpos}), so as to be able to speak of density of pedestrians at least in the sense of Radon-Nikodym theorem. Notice that if, on the one hand, measuring the space occupancy by $\mu_n$ makes perfectly sense in a macroscopic frame, claiming as a modeling assumption that pedestrians can be physically described by a continuously distributed density may be, on the other hand, more hardly acceptable and more difficult to justify. This issue is actually common to many other nonclassical systems, for which the concepts of mass and mass density are not as straightforward as in standard continuum mechanics. The theory we have developed in this paper provides, as a by-product, a possible methodology that, starting from reasonably meaningful macroscopic concepts, may help in gaining, at least in the abstract, a pointwise density for the system under consideration.

The existence of a density is useful for the numerical treatment of the equations of the model, indeed approximating $\rho_n$ as a mapping on $\R^d$ is undoubtedly more practical than approximating $\mu_n$ as a set mapping defined on some $\sigma$-algebra of $\R^d$. In the paper we have proposed a numerical scheme to discretize, at each time step, the push forward of the $\mu_n$'s over a suitable space grid in the domain $\Omega$. In addition, we have studied the error produced by the numerical solution both in a single step of push forward and on the true density $\rho_n$, on the basis of the level of refinement of the space grid compared to the discrete time step $\Delta{t}$. The scheme has proved to be robust and efficient in providing potentially accurate approximations of the true solution (cf. Theorems \ref{theo:onestepstab}, \ref{theo:multistepstab}), as well as in preserving some interesting theoretical features of the latter, e.g., nonnegativity, integrability, and local in time boundedness.

As mentioned above, the existence of a density is often also helpful in order to look conceptually at the system from a point of view close to that of classical continuum mechanics, even when, like in the case of pedestrian flow, the dynamics is mainly dictated by nonclassical evolution rules. In the second part of the work (cf. Sect. \ref{sect:modeling}) we have been concerned with the application of the time-evolving measures theory to the modeling of crowd motion. In particular, we have proposed a form of the velocity of pedestrians inspired by two main guidelines, somehow common to several microscopic models of pedestrian flow already available in the specialized literature: (i) The will of pedestrians to reach specific targets within the walking area, bypassing at the same time obstacles possibly present along their paths (desired velocity); (ii) The tendency of pedestrians, at least in normal, i.e. no panic, conditions, to look for uncrowded areas in their neighborhoods, in order not to stay too close to one another (interaction velocity). Notice that the simultaneous presence of Eulerian and Lagrangian aspects in the description of the system allows to duly incorporate microscopic viewpoints, namely the motion mapping $\gamma_n$ and consequently the construction of the velocity $v_n$, within a global macroscopic perspective on the system provided by the measures $\mu_n$.

\begin{figure}[t]
\centering
\includegraphics{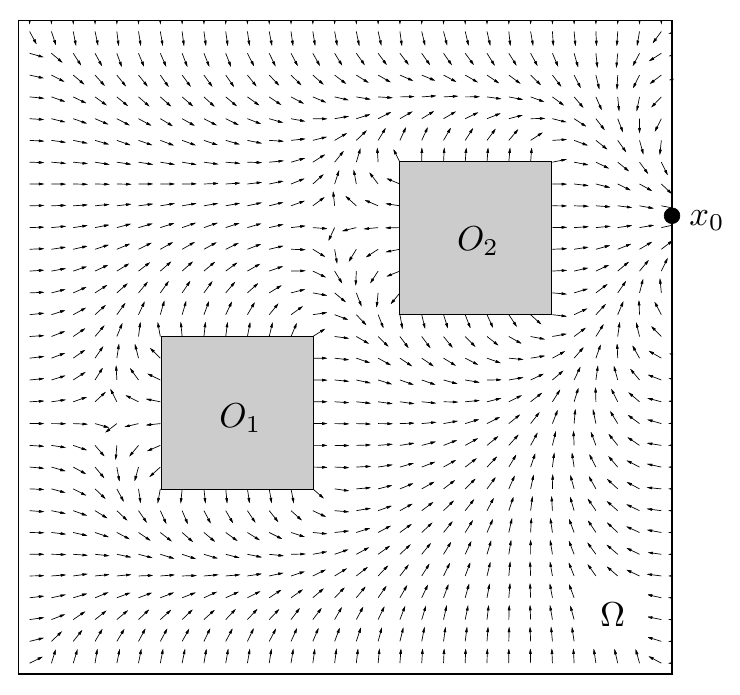}
\caption{Desired velocity in presence of two obstacles $O_1$, $O_2$. The field is generated by a potential attracting pedestrians toward the target $x_0$.}
\label{fig:desvel_pot}
\end{figure}

At present, our treatment of the desired velocity is mainly illustrative. In particular, we are assuming that pedestrians aim at one specific target, for instance an aggregation point or a door along the wall, that they can reach either directly or by possibly avoiding, in a simple manner, a single intermediate obstacle within the walking area (cf. Fig. \ref{fig:obstacle}). However, it is plain that the most interesting problems concern the flow of people in areas scattered with obstacles, like e.g., pillars, bottlenecks, narrow passages, therefore we are currently scheduling a forthcoming work in which we will further investigate this modeling aspect, relying on the essentially geometrical nature of the desired velocity field. As an anticipation, we observe that it is for instance possible to generate smooth fields from a potential $u:\Omega\to\R$, which attracts pedestrians toward a target $x_0\in\partial\Omega$, by solving the elliptic problem
$$	\left\{
	\begin{array}{rcll}
		-\Delta{u} & = & 0 & \text{in\ } \Omega \\
		u & = & g & \text{on\ } \partial\Omega,
	\end{array}
	\right.
$$
where $g$ vanishes on all external and internal boundaries of $\Omega$ but the one containing $x_0$, and then setting $v_d(x)=\nabla{u}(x)$ (see Fig. \ref{fig:desvel_pot}).

Concerning the interaction velocity, the main characteristic is that it accounts for nonlocal effects within a certain finite-radius neighborhood around pedestrians. In practice, each individual is assumed to have her/his walking program influenced, in average, by all people comprised in some neighborhood of interaction, in such a way that she/he can possibly agree to deviate locally from her/his preferred path in order to bypass crowded areas, looking preferentially for free space. In the present work we have shown that a nonlocal average structure of the interaction velocity fits well in the proposed time-evolving measures theoretical framework (cf. Proposition \ref{prop:mod-abscont}). From the modeling side, this proposal deserves further investigation and development in order to provide a simple but accurate description of the interactions among pedestrians, which will be in turn the object of the above-announced subsequent work.

\bibliographystyle{plain}
\bibliography{/dati/Ricerca/Bibliografie/Pedestrian/pedestrian-biblio,/dati/Ricerca/Bibliografie/Traffico/traffic-biblio}
\end{document}